\documentclass[12pt]{article}
\usepackage[margin=1in]{geometry}
\usepackage{amsmath,amssymb,amsthm,mathtools,bm}
\usepackage[T1]{fontenc}
\usepackage{lmodern}
\usepackage{microtype}
\usepackage{setspace}
\usepackage{authblk}
\usepackage[round,authoryear]{natbib}
\usepackage{placeins}
\usepackage{hyperref}

\hypersetup{colorlinks=true,linkcolor=blue,citecolor=blue,urlcolor=blue}

\newtheorem{theorem}{Theorem}
\newtheorem{proposition}[theorem]{Proposition}

\theoremstyle{definition}
\newtheorem{remark}[theorem]{Remark}

\newcommand{\NN}{\mathbb{N}}
\newcommand{\state}{\Omega_n^{(m)}}

\title{Selection Mechanisms, Stationary Distributions, and Reversibility in Multiallelic Moran Models}
\author[1,3]{Dan Braha}
\author[2,3]{Marcus A. M. de Aguiar}
\affil[1]{University of Massachusetts Dartmouth, Dartmouth, Massachusetts, United States of America}
\affil[2]{Gleb Wataghin Institute of Physics, State University of Campinas (UNICAMP), Campinas, Brazil}
\affil[3]{New England Complex Systems Institute, Cambridge, Massachusetts, United States of America}
\date{}

\begin{document}
\maketitle

\begin{abstract}
The Moran process with selection and recurrent mutation is a classical model in population genetics, yet how the placement of selection within the update rule shapes the stationary distribution has received little attention. We study a finite, well-mixed haploid population of constant size $n$ with $m$ labeled alleles, parent-independent mutation, and allele-specific fitnesses. Within this common framework we compare three Moran update kernels that differ only in the stage at which selection acts: during reproduction, when the offspring copies one of two sampled parents (Scheme~I); through fitness-biased mate choice, followed by neutral copying (Scheme~II); and at death, so that fitter individuals are less likely to be replaced (Scheme~III). Although all three favor fitter alleles, they define different Markov chains. For two alleles, each scheme reduces to a birth--death chain and admits an exact stationary law, but the three laws differ. For $m\ge 3$, the placement of selection becomes decisive: Schemes~I and~II are generally nonreversible when fitnesses are unequal, so no detailed-balance product form exists, whereas Scheme~III remains reversible for every $m$ and has a closed stationary distribution---a Dirichlet--multinomial core modified by an explicit fitness factor. We further show that all three mechanisms can act simultaneously in the two-allele case without losing exact solvability, and we derive weak-selection expansions that make explicit how small fitness differences tilt the neutral beta-binomial and Dirichlet--multinomial benchmarks. Together, these results clarify when neutral stationary structure survives the introduction of selection and when multiallelic Moran dynamics become genuinely nonreversible.
\end{abstract}

\section{Introduction}
The Moran process is one of the standard finite-population models in population genetics. Its overlapping-generation structure lets allele counts evolve as a finite Markov chain, which has made the model a central tool for exact analysis of drift, mutation, selection, and fixation. In the neutral model with multiple alleles and parent-independent mutation, the stationary law of the labeled allele-count vector is Dirichlet--multinomial. This paper asks how that exact finite-population stationary structure changes once selection is introduced---and, in particular, how the answer depends on where selection enters the update rule.

We study a finite, well-mixed haploid population of constant size $n$ with $m$ labeled alleles, recurrent parent-independent mutation, and an allele-specific fitness coefficient for each allele. Within this common framework we compare three Moran update kernels that differ only in the stage at which selection acts: during reproduction, through fitness-biased mate choice, or at death. The three kernels are defined precisely in Section~\ref{sec:model}.

The contribution is not to introduce selection into the Moran process, but to show that the stage at which selection enters the update is mathematically consequential. Although all three schemes favor fitter alleles, they define different finite Markov chains, with different stationary distributions and different reversibility properties. For two alleles, each scheme reduces to a birth--death chain whose stationary law follows from local detailed balance; the three laws are nonetheless distinct. For three or more alleles the contrast sharpens: when fitnesses differ, selection during reproduction and fitness-biased mate choice generally make the process nonreversible, so no detailed-balance product form is available, whereas selection at death remains reversible for arbitrary $m$ and yields a closed stationary distribution---the neutral Dirichlet--multinomial weight multiplied by an explicit fitness factor. The placement of selection thus governs both the form of the two-allele stationary law and the reversibility of the multiallelic process.

The paper is organized as follows. Section~\ref{sec:literature} places the contribution in the context of previous work. Section~\ref{sec:model} introduces the model and derives the transition kernels for the three schemes. Section~\ref{sec:reversibility-scope} gives the reversibility classification and explains why the two-allele case is the natural exact focus for Schemes~I and~II. Sections~\ref{sec:scheme-I-two} and~\ref{sec:scheme-II-two} derive the exact two-allele stationary laws for Schemes~I and~II. Section~\ref{sec:scheme-III} derives the reversible stationary distribution for Scheme~III with multiple alleles and then specializes it to two alleles. Section~\ref{sec:combined-two} shows that all three mechanisms can be combined in the two-allele case without losing an exact stationary law. Section~\ref{sec:weak-selection} analyzes the weak-selection regime, showing how small fitness differences perturb the neutral beta-binomial and Dirichlet--multinomial stationary laws. Section~\ref{sec:numerical-illustrations} illustrates the exact stationary laws by plotting allele-count distributions for representative mutation rates and selection strengths. Section~\ref{sec:discussion} concludes the paper. Technical details are collected in the appendices.

\section{Relation to previous work and contribution}
\label{sec:literature}

Stochastic models of selection in finite populations have developed along two closely related approaches. One is based on the Wright--Fisher model, where selection is often studied through diffusion approximations. The other is based on the Moran model, where overlapping generations lead to a finite Markov chain that can sometimes be analyzed exactly. Classical work by \citet{moran1958random,moran1958selection,moran1959survival,moran1962book} already contained several themes that continue to organize the subject, including absorption, fixation, and survival in the absence of recurrent mutation, as well as mutation--selection balance when mutation persists. Related work by \citet{karlinmcgregor1962}, \citet{ewens1963}, and \citet{cannings1974haploid} clarified the structure of finite-population genetic Markov chains, while \citet{kimura1962fixation} and \citet{crowkimura1970} helped establish diffusion methods as a central language for selection in the Wright--Fisher setting. The present paper belongs to the exact finite-population side of this literature. Its focus is not fixation or absorption times, but stationary distributions under recurrent mutation and selection.

For multiallelic loci under balancing selection, much of the classical literature uses diffusion or asymptotic methods adapted to particular biological mechanisms. Important examples include Wright's treatment of gametophytic self-incompatibility \citep{wright1939selfsterility}, the analysis of neutral, mutually heterotic, and deleterious isoalleles by \citet{kimuracrow1964}, and Watterson's diffusion analysis of heterosis versus neutrality \citep{watterson1977}. Later work on self-incompatibility and major histocompatibility complex loci further developed genealogical and diffusion-based descriptions of balancing selection \citep{takahata1990,takahata1992mhc,vekemansslatkin1994,uyenoyama2003}. This literature established that multiallelic selection can maintain deep polymorphism and distinctive allele frequency structure. It also illustrates a recurring methodological difficulty: exact finite-population stationary laws are rare, and many available expressions depend on symmetry assumptions, diffusion approximations, or biological restrictions built into particular models, such as \(S\)-locus compatibility rules in self-incompatibility or MHC overdominance assumptions motivated by heterozygote advantage. 

The Moran model is attractive in this setting because it can preserve exact finite-population structure that is difficult to obtain in the Wright--Fisher Markov chain. In the standard continuous-time two-allele Moran model with fecundity selection, the allele count forms a one-dimensional birth--death Markov chain, so fixation probabilities, mean absorption times, and, with recurrent mutation, stationary distributions can often be written explicitly. A large literature has used this tractability to study weak selection, frequency-dependent fitness, and mutation--selection balance in finite populations; representative examples include \citet{taylor2004finite}, \citet{fudenberg2006finite}, \citet{ohtsuki2007one}, \citet{antal2009mutation}, \citet{wu2010universality}, and \citet{trotter2013models}, with more recent extensions to boundary-mutation equilibria, nonlinear selection, and large-population asymptotics \citep{voglmikula2021,cloezcorujo2022}. This work shows that the details of the update rule can affect abundance, fixation, and equilibrium behavior, but its focus is usually two-allele dynamics, weak-selection criteria, fixation, or asymptotic approximations rather than exact stationary distributions on labeled multiallelic count vectors with allele-specific fitness coefficients.

The closest point of comparison is the multiallelic Moran analysis of \citet{muirheadwakeley2009}. They developed a continuous-time Moran approach to multiallelic selection with parent-independent mutation and derived reversible stationary formulas for \emph{exchangeable} selection models, in which relabeling alleles does not change selective outcomes, so that selection depends on allele frequency rather than allele identity. Their state is the allele-frequency spectrum---the number of alleles in each frequency class---and their formulas give expected spectra under several exchangeable models. Their analysis is close in spirit to ours, but it also marks the key distinction. Here the state is a labeled allele-count vector and selection is allele-specific through fixed coefficients \(S_1,\ldots,S_m\), so that, unless all fitnesses are equal, relabeling alleles changes the model. We therefore ask a different question: how does the stationary structure change when the same labeled, allele-specific fitnesses are inserted at different stages of the Moran update?

A related body of work studies Moran models with selection from an ancestral or genealogical perspective rather than by deriving stationary distributions for allele-count vectors. In the two-type setting, mutation and selection have been analyzed from this backward viewpoint using ancestral selection graph, common ancestor, and particle representation methods, which trace ancestral lines from a present or equilibrium population and relate them to forward type distributions or to fixation probabilities and the selective events that shape them \citep{baakebialowons2008,kluthbaake2013}. This work demonstrates the analytic power of the Moran framework, but its emphasis is genealogical structure, fixation, and particle representations---mainly in two-type continuous-time models where selection acts through type-dependent reproduction---rather than stationary laws for labeled multiallelic counts.

There is also a broader literature showing that the placement of selection within a finite-population update rule can affect evolutionary outcomes. In evolutionary graph theory and related models with population structure, changing the order in which reproduction, replacement, and fitness bias enter the update can change fixation and abundance behavior \citep{lieberman2005evolutionary,nowak2010structured, constablemckane2015}. More generally, \citet{molinaearn2018} formalized a broad class of finite-population selection processes to emphasize that evolutionary outcomes depend on the rule that converts fitness differences into transition probabilities. This literature is usually concerned with fixation, abundance, or structured populations rather than labeled multiallelic stationary laws, but its message is directly relevant here: two update rules may favor the same fitter alleles yet define different Markov chains, and therefore different stationary distributions.

Models of assortative mating and self-incompatibility provide another related but more biologically specialized context. \citet{etheridgelemaire2011}, for example, study a Moran-type multilocus model with weak assortative mating, strong recombination, and low mutation, and derive a diffusion approximation for allele frequencies. Work on haploid mating types and gametophytic self-incompatibility analyzes rare-type advantage, invasion, extinction, and the number of types maintained in finite populations \citep{czupponrogers2019,czupponbilliard2022}. These papers show how mating rules and compatibility constraints can create selection-like effects in finite populations, but they target diffusion limits, invasion probabilities, type loss, or specific self-incompatibility systems rather than exact stationary distributions for labeled alleles under alternative placements of selection.

Against this background, the present paper isolates a question that the works above do not address directly: holding the population, the mutation scheme, and the allele-specific fitnesses fixed, how does the stationary distribution depend on the stage of the update at which selection acts? We compare three placements---selection during reproduction, through fitness-biased mate choice, and at death---within a single finite, well-mixed haploid Moran framework with parent-independent mutation. This turns the general observation above, that update rules favoring the same alleles may nonetheless define different Markov chains, into an exact and explicit statement for a specific trio of mechanisms. To our knowledge, this trio of stage-dependent kernels has not previously been formulated and analyzed together in a labeled multiallelic setting; the resulting classification of their stationary and reversibility structure is developed in the sections that follow.

\section{Model and notation}\label{sec:model}
We consider a fixed-size haploid population of $n\ge 2$ individuals, or gametes, with a single genetic locus having $m\ge 2$ possible alleles, labeled $1,2,\ldots,m$. The state of the population is the labeled allele-count vector
\[
 k=(k_1,\dots,k_m)\in \state:=\Big\{k\in \NN_0^m:\ \sum_{i=1}^m k_i=n\Big\}.
\]
Here $k_i$ is the number of individuals carrying allele $i$. To each allele $i$ we associate a positive fitness coefficient $S_i>0$. Fitness is allele-specific and fixed throughout the dynamics.

Each update replaces one individual by one offspring. The paper compares three ways in which selection can enter this replacement event.

In Scheme~I, a focal individual is chosen uniformly for replacement. Suppose the focal individual has allele $j$. A second parent with allele $r$ is chosen uniformly from the remaining $n-1$ individuals. The two parents generate an offspring, which copies the focal allele with probability $S_j/(S_j+S_r)$ and copies the second parent's allele with probability $S_r/(S_j+S_r)$. The copied allele may then mutate before the offspring replaces the focal individual. Selection therefore acts during the copying step of reproduction.

In Scheme~II, the focal individual is again chosen uniformly for replacement. The second parent is now chosen from the remaining population with probability proportional to fitness. Once the focal individual and the second parent have been chosen, the offspring copies either parental allele with probability $1/2$, and mutation then acts on the copied allele. Selection therefore acts through mate choice rather than through the copying probability.

In Scheme~III, selection acts at death. The individual to be replaced is chosen with probability proportional to inverse fitness, so fitter individuals are less likely to be removed. Conditional on the individual chosen for replacement, reproduction is neutral apart from mutation: a second parent is chosen uniformly from the remaining population, the offspring copies either parental allele with probability $1/2$, and mutation then acts on the copied allele.

We now introduce notation used to define the model and its transition kernels. For $i\neq j$, let $e_i$ denote the $i$th coordinate vector. The move
\[
 k\longrightarrow k+e_i-e_j
\]
means that one type-$j$ individual is replaced by a type-$i$ individual.

Mutation is parent-independent. Thus, for $r\neq i$,
\[
\mu_{r\to i}=u_i,
\qquad u_i>0,
\]
with
\[
\Lambda:=\sum_{i=1}^m u_i<1,
\qquad
\mu_{i\to i}=1-\Lambda+u_i.
\]
The term $u_i$ is the baseline parent-independent probability that the offspring ends as type $i$. Thus, after copying an allele different from $i$, the offspring ends as type $i$ with probability $u_i$, whereas after copying allele $i$ itself, it ends as type $i$ with probability $1-\Lambda+u_i$.

Throughout the paper, $Q$ denotes one-step transition probabilities. We use superscripts to distinguish the three schemes: $Q^{(I)}$, $Q^{(II)}$, and $Q^{(III)}$. The diagonal probability is defined by
\[
Q^{(s)}(k,k)=1-\sum_{i\neq j}Q^{(s)}\bigl(k,k+e_i-e_j\bigr),
\qquad s\in\{I,II,III\}.
\]
We also use the rising factorial notation
\[
(a)_r:=\frac{\Gamma(a+r)}{\Gamma(a)},
\qquad r=0,1,2,\ldots .
\]
Finally, define the neutral reproduction coefficient \(b\) and the associated scaled mutation parameters \(\alpha_i\) by
\begin{equation}\label{eq:b-alpha-def}
 b:=\frac{1-\Lambda}{2(n-1)},
 \qquad
 \alpha_i:=\frac{u_i}{b}=\frac{2(n-1)u_i}{1-\Lambda}.
\end{equation}
The factor $2$ in \eqref{eq:b-alpha-def} appears because, in the neutral two-parent reproduction step, the offspring copies the second parent with probability $1/2$. With this parametrization, the neutral stationary law has Dirichlet--multinomial form; see Appendix~\ref{app:neutral-DM}.

\subsection{Scheme I: selection during reproduction}\label{subsec:scheme-I-kernel}
In Scheme~I, the individual to be replaced is chosen uniformly. Suppose this \textit{focal} individual has allele $j$. A second parent is then chosen uniformly from the remaining $n-1$ individuals. If the second parent has allele $r$, the offspring copies allele $j$ with probability
\[
\frac{S_j}{S_j+S_r}
\]
and copies allele $r$ with probability
\[
\frac{S_r}{S_j+S_r}.
\]
Mutation then acts on the copied allele. The transition probabilities for Scheme~I are given in the following proposition.

\begin{proposition}
\label{prop:scheme-I-kernel}
For $i\neq j$ and $k\in\state$,
\begin{equation}\label{eq:QI-kernel}
Q^{(I)}\bigl(k,k+e_i-e_j\bigr)
=\frac{k_j}{n}\,T^{(I)}_{i\mid j}(k),
\end{equation}
where
\begin{equation}\label{eq:TI-final}
T^{(I)}_{i\mid j}(k)
=u_i+b\,s_{ij}k_i,
\qquad
s_{ij}:=\frac{2S_i}{S_i+S_j}.
\end{equation}
Equivalently,
\begin{equation}\label{eq:QI-kernel-expanded}
Q^{(I)}\bigl(k,k+e_i-e_j\bigr)
=\frac{k_j}{n}\left(u_i+\frac{1-\Lambda}{n-1}\frac{S_i}{S_i+S_j}\,k_i\right).
\end{equation}
\end{proposition}

\begin{proof}
We derive the transition probability by conditioning first on the type of the focal individual and then on the type of the second parent. The move
\[
k\longrightarrow k+e_i-e_j,\qquad i\neq j,
\]
occurs exactly when the focal individual has type $j$ and the offspring that replaces it has type $i$. Since the focal individual is chosen uniformly,
\[
Q^{(I)}\bigl(k,k+e_i-e_j\bigr)
=\frac{k_j}{n}\,T^{(I)}_{i\mid j}(k),
\]
where $T^{(I)}_{i\mid j}(k)$ is the conditional probability that the replacement offspring is of type $i$, given that the focal individual is of type $j$.

It remains to compute $T^{(I)}_{i\mid j}(k)$. Conditional on a type-$j$ focal individual, the second parent is chosen uniformly from the remaining $n-1$ individuals. If the copied allele is already $i$, the offspring ends as type $i$ with probability $1-\Lambda+u_i$: either no mutation occurs, or mutation occurs and returns allele $i$. If the copied allele is not $i$, the offspring ends as type $i$ only through mutation, with probability $u_i$.

This conditional probability is obtained by separating the possible second parent into three mutually exclusive cases:

First, the second parent may have type $i$. This occurs with probability $k_i/(n-1)$. In that case the offspring can end as type $i$ in two ways: it copies the focal type-$j$ allele and mutates to $i$, or it copies the type-$i$ second parent and remains type $i$. The contribution is
\[
A
=
\frac{k_i}{n-1}
\left[
\frac{S_j}{S_j+S_i}u_i
+
\frac{S_i}{S_j+S_i}(1-\Lambda+u_i)
\right].
\]

Second, the second parent may also have type $j$. There are $k_j-1$ such possible partners. Both parental alleles are then $j$, so regardless of which parent is copied, the offspring must mutate to $i$. The contribution is
\[
B
=
\frac{k_j-1}{n-1}\,u_i .
\]

Third, the second parent may have type $r\in\mathcal R_{ij}$, where $\mathcal R_{ij}:=\{1,\ldots,m\}\setminus\{i,j\}$. For a fixed such $r$, this occurs with probability $k_r/(n-1)$. If the offspring copies the focal type-$j$ allele, it must mutate to $i$; if it copies the type-$r$ partner, it must also mutate to $i$. Hence the contribution from all $r\in\mathcal R_{ij}$ is
\[
C
=
\sum_{r\in\mathcal R_{ij}}
\frac{k_r}{n-1}
\left[
\frac{S_j}{S_j+S_r}u_i
+
\frac{S_r}{S_j+S_r}u_i
\right]
=
\frac{1}{n-1}\sum_{r\in\mathcal R_{ij}}k_r u_i .
\]

Adding the three mutually exclusive contributions gives
\begin{align*}
T^{(I)}_{i\mid j}(k)
&=A+B+C \\
&=
\frac{1}{n-1}\Bigg\{
k_i\left[
\frac{S_j}{S_j+S_i}u_i
+
\frac{S_i}{S_j+S_i}(1-\Lambda+u_i)
\right]
+(k_j-1)u_i
+\sum_{r\in\mathcal R_{ij}}k_ru_i
\Bigg\}.
\end{align*}
The bracket involving the type-$i$ second parent simplifies to
\[
\frac{S_j}{S_j+S_i}u_i
+
\frac{S_i}{S_j+S_i}(1-\Lambda+u_i)
=
u_i+\frac{S_i}{S_i+S_j}(1-\Lambda).
\]
Therefore
\begin{align*}
T^{(I)}_{i\mid j}(k)
&=
\frac{1}{n-1}
\left\{
\left(k_i+k_j-1+\sum_{r\in\mathcal R_{ij}}k_r\right)u_i
+
k_i\frac{S_i}{S_i+S_j}(1-\Lambda)
\right\}.
\end{align*}
Since
\[
k_i+k_j+\sum_{r\in\mathcal R_{ij}}k_r=n,
\]
the coefficient of $u_i$ is $n-1$. Hence
\[
T^{(I)}_{i\mid j}(k)
=
u_i+\frac{k_i}{n-1}\frac{S_i}{S_i+S_j}(1-\Lambda)
=
u_i+b\,s_{ij}k_i,
\]
where \(b\) is defined in \eqref{eq:b-alpha-def} and $s_{ij}=2S_i/(S_i+S_j)$. Substituting this expression into
\[
Q^{(I)}\bigl(k,k+e_i-e_j\bigr)
=\frac{k_j}{n}T^{(I)}_{i\mid j}(k)
\]
proves \eqref{eq:TI-final} and \eqref{eq:QI-kernel-expanded}.
\end{proof}

It is useful to rewrite the transition probabilities in
\eqref{eq:QI-kernel}--\eqref{eq:TI-final} in a form parallel to the neutral
case. Define the effective mutation parameters
\begin{equation}\label{eq:tilde-alpha-I}
\widetilde\alpha_{ij}:=\frac{u_i}{b s_{ij}}
=\frac{\alpha_i}{s_{ij}}
=\frac{2(n-1)u_i}{1-\Lambda}\frac{S_i+S_j}{2S_i}.
\end{equation}
Then
\begin{equation}\label{eq:QI-alpha-form}
Q^{(I)}\bigl(k,k+e_i-e_j\bigr)
=\frac{k_j}{n}\,b s_{ij}\bigl(k_i+\widetilde\alpha_{ij}\bigr).
\end{equation}
In the neutral case $S_1=\cdots=S_m$, one has $s_{ij}=1$ and $\widetilde\alpha_{ij}=\alpha_i$.

\subsection{Scheme II: fitness-biased mate choice}\label{subsec:scheme-II-kernel}
In Scheme~II, the focal individual to be replaced is again chosen uniformly. Suppose its allele is $j$. The second parent is then chosen from the remaining population with probability proportional to fitness. Once the focal individual and the second parent have been chosen, the offspring copies either parental allele with probability $1/2$, and mutation then acts on the copied allele.

Define the fitness-weighted population total
\begin{equation}\label{eq:xi-def}
\xi(k):=\sum_{\ell=1}^m S_\ell k_\ell.
\end{equation}
If the focal allele is $j$, the total fitness available among possible second parents is $\xi(k)-S_j$. The transition probabilities for Scheme~II are given in the following proposition.

\begin{proposition}
\label{prop:scheme-II-kernel}
For $i\neq j$ and $k\in\state$,
\begin{equation}\label{eq:QII-kernel}
Q^{(II)}\bigl(k,k+e_i-e_j\bigr)
=\frac{k_j}{n}\,T^{(II)}_{i\mid j}(k),
\end{equation}
where
\begin{equation}\label{eq:TII-final}
T^{(II)}_{i\mid j}(k)
=u_i+\frac{(1-\Lambda)S_i k_i}{2\bigl(\xi(k)-S_j\bigr)}.
\end{equation}
\end{proposition}

\begin{proof}
As in Scheme~I, the move
\[
k\longrightarrow k+e_i-e_j,\qquad i\neq j,
\]
occurs when the focal individual has type $j$ and the replacement offspring ultimately has type $i$. The focal individual is chosen uniformly, so
\[
Q^{(II)}\bigl(k,k+e_i-e_j\bigr)
=\frac{k_j}{n}\,T^{(II)}_{i\mid j}(k),
\]
where $T^{(II)}_{i\mid j}(k)$ is computed conditional on the focal individual being of type $j$.

In Scheme~II, selection acts in the choice of the second parent. Once the type-$j$ focal individual has been removed from the pool of possible second parents, the available total fitness is
\[
\xi(k)-S_j.
\]
Thus a type-$r$ second parent is chosen with probability
\[
\frac{S_r\bigl(k_r-\delta_{rj}\bigr)}{\xi(k)-S_j}
\]
where $\delta_{rj}$ is the Kronecker delta.
After this second parent has been chosen, copying is neutral: the offspring copies either parent with probability $1/2$. As before, a copied allele $i$ gives type $i$ with probability $1-\Lambda+u_i$, while a copied allele different from $i$ gives type $i$ only with probability $u_i$.

We now decompose $T^{(II)}_{i\mid j}(k)$ into three cases, according to the type of the second parent:

First, the second parent may have type $i$. This has probability $S_ik_i/(\xi(k)-S_j)$. Conditional on this event, the offspring can end as type $i$ either by copying the focal type-$j$ allele and mutating to $i$, or by copying the type-$i$ second parent and remaining type $i$. This gives
\[
A
=
\frac{S_ik_i}{\xi(k)-S_j}
\left[
\frac12 u_i+\frac12(1-\Lambda+u_i)
\right].
\]

Second, the second parent may have type $j$. This has probability $S_j(k_j-1)/(\xi(k)-S_j)$. Both parental alleles are then $j$, so the offspring must mutate to $i$. The contribution is
\[
B
=
\frac{S_j(k_j-1)}{\xi(k)-S_j}\,u_i .
\]

Third, the second parent may have type $r\in\mathcal R_{ij}$, where $\mathcal R_{ij}:=\{1,\ldots,m\}\setminus\{i,j\}$. For each such $r$, neither parental allele is $i$: one parent has type $j$ and the other has type $r$. Therefore the offspring ends as type $i$ only by mutation, regardless of which parent is copied. The total contribution is
\[
C
=
\sum_{r\in\mathcal R_{ij}}
\frac{S_rk_r}{\xi(k)-S_j}
\left[\frac12u_i+\frac12u_i\right]
=
\sum_{r\in\mathcal R_{ij}}
\frac{S_rk_r}{\xi(k)-S_j}u_i .
\]

Adding these contributions,
\begin{align*}
T^{(II)}_{i\mid j}(k)
&=A+B+C \\
&=
\frac{S_ik_i}{\xi(k)-S_j}
\left(u_i+\frac{1-\Lambda}{2}\right)
+\frac{S_j(k_j-1)}{\xi(k)-S_j}u_i
+\sum_{r\in\mathcal R_{ij}}\frac{S_rk_r}{\xi(k)-S_j}u_i \\
&=
\frac{1}{\xi(k)-S_j}
\left\{
\left[S_ik_i+S_j(k_j-1)+\sum_{r\in\mathcal R_{ij}}S_rk_r\right]u_i
+\frac{1-\Lambda}{2}S_ik_i
\right\}.
\end{align*}
The expression in square brackets is the total fitness of the available second-parent pool:
\[
S_ik_i+S_j(k_j-1)+\sum_{r\in\mathcal R_{ij}}S_rk_r
=
\xi(k)-S_j.
\]
Consequently
\[
T^{(II)}_{i\mid j}(k)
=
u_i+\frac{(1-\Lambda)S_i k_i}{2\bigl(\xi(k)-S_j\bigr)},
\]
which proves \eqref{eq:TII-final}.
\end{proof}

To make the comparison with Scheme~I explicit, rewrite
\eqref{eq:TII-final} in the form
\[
T^{(II)}_{i\mid j}(k)=u_i+b\,m_{ij}(k)k_i,
\]
where
\begin{equation}\label{eq:m-ij-def}
m_{ij}(k):=\frac{S_i(n-1)}{\xi(k)-S_j}.
\end{equation}
Here \(m_{ij}(k)\) plays the role, for Scheme~II, that the constant
factor \(s_{ij}\) plays for Scheme~I. The important difference is that
\(m_{ij}(k)\) is state-dependent: through \(\xi(k)\), it depends on the
full current allele-count vector.

\subsection{Scheme III: selection at death}\label{subsec:scheme-III-kernel}
In Scheme~III, selection acts in the choice of the individual to be replaced. Individuals with larger fitness are less likely to die. We model this by choosing a type-$j$ individual for replacement with probability proportional to $S_j^{-1}k_j$.

Define
\begin{equation}\label{eq:nu-def}
\nu(k):=\sum_{\ell=1}^m S_\ell^{-1}k_\ell.
\end{equation}
Then the probability that the individual chosen for replacement has allele $j$ is
\begin{equation}\label{eq:death-prob}
\rho_j(k):=\frac{S_j^{-1}k_j}{\nu(k)}=\frac{k_j}{S_j\nu(k)}.
\end{equation}
After this death-selection step, reproduction is neutral: conditional on the focal allele being $j$, the replacement mechanism has the same conditional probability as the neutral version of Scheme~I. The transition probabilities for Scheme~III are given in the following proposition.

\begin{proposition}
\label{prop:scheme-III-kernel}
For $i\neq j$ and $k\in\state$,
\begin{equation}\label{eq:QIII-kernel}
Q^{(III)}\bigl(k,k+e_i-e_j\bigr)
=\frac{k_j}{S_j\nu(k)}\,T^{(0)}_{i\mid j}(k),
\end{equation}
where
\begin{equation}\label{eq:T0-final}
T^{(0)}_{i\mid j}(k)=u_i+b k_i.
\end{equation}
Equivalently,
\begin{equation}\label{eq:QIII-alpha-form}
Q^{(III)}\bigl(k,k+e_i-e_j\bigr)
=\frac{b}{\nu(k)}\frac{k_j}{S_j}\bigl(k_i+\alpha_i\bigr).
\end{equation}
\end{proposition}

\begin{proof}
In Scheme~III, selection acts only in the death step. A type-$j$ individual is chosen for replacement with probability
\[
\rho_j(k)
=
\frac{S_j^{-1}k_j}{\nu(k)}
=
\frac{k_j}{S_j\nu(k)}.
\]
Conditional on this event, reproduction is neutral apart from mutation. Thus
\[
Q^{(III)}\bigl(k,k+e_i-e_j\bigr)
=
\frac{k_j}{S_j\nu(k)}\,T^{(0)}_{i\mid j}(k),
\]
where $T^{(0)}_{i\mid j}(k)$ is the neutral conditional probability that the offspring replacing the focal type-$j$ individual ends as type $i$.

We compute $T^{(0)}_{i\mid j}(k)$ by conditioning on the second parent. The second parent is chosen uniformly from the remaining $n-1$ individuals, and the offspring copies either parent with probability $1/2$. If the copied allele is $i$, the offspring ends as type $i$ with probability $1-\Lambda+u_i$; if the copied allele is not $i$, it ends as type $i$ only through mutation, with probability $u_i$.

This conditional probability is obtained by separating the possible second parent into three mutually exclusive cases:

First, the second parent may have type $i$. This occurs with probability $k_i/(n-1)$. The offspring then ends as type $i$ either by copying the focal type-$j$ allele and mutating to $i$, or by copying the type-$i$ second parent and remaining type $i$. The contribution is
\[
A
=
\frac{k_i}{n-1}
\left[
\frac12u_i+\frac12(1-\Lambda+u_i)
\right].
\]

Second, the second parent may have type $j$. This occurs with probability $(k_j-1)/(n-1)$. Both parental alleles are $j$, so the offspring must mutate to $i$. The contribution is
\[
B
=
\frac{k_j-1}{n-1}\,u_i .
\]

Third, the second parent may have type $r\in\mathcal R_{ij}$, where $\mathcal R_{ij}:=\{1,\ldots,m\}\setminus\{i,j\}$. For these partners, neither parental allele is $i$, and the offspring must mutate to $i$ regardless of which parent is copied. The total contribution is
\[
C
=
\sum_{r\in\mathcal R_{ij}}
\frac{k_r}{n-1}
\left[\frac12u_i+\frac12u_i\right]
=
\frac{1}{n-1}\sum_{r\in\mathcal R_{ij}}k_ru_i .
\]

Adding the contributions gives
\begin{align*}
T^{(0)}_{i\mid j}(k)
&=A+B+C \\
&=
\frac{1}{n-1}
\left\{
k_i\left(u_i+\frac{1-\Lambda}{2}\right)
+(k_j-1)u_i
+\sum_{r\in\mathcal R_{ij}}k_ru_i
\right\} \\
&=
\frac{1}{n-1}
\left\{
\left(k_i+k_j-1+\sum_{r\in\mathcal R_{ij}}k_r\right)u_i
+\frac{1-\Lambda}{2}k_i
\right\}.
\end{align*}
Since
\[
k_i+k_j+\sum_{r\in\mathcal R_{ij}}k_r=n,
\]
this becomes
\[
T^{(0)}_{i\mid j}(k)
=
u_i+\frac{1-\Lambda}{2(n-1)}k_i
=
u_i+bk_i.
\]
Substituting into the transition probability gives \eqref{eq:QIII-kernel}. Finally, using $u_i=b\alpha_i$ gives
\[
Q^{(III)}\bigl(k,k+e_i-e_j\bigr)
=
\frac{b}{\nu(k)}\frac{k_j}{S_j}\bigl(k_i+\alpha_i\bigr),
\]
which is \eqref{eq:QIII-alpha-form}.
\end{proof}

\section{Reversibility and the focus on two alleles}\label{sec:reversibility-scope}
The three schemes differ sharply in their reversibility properties. The following theorem summarizes the classification needed for the rest of the paper. 

\begin{theorem}
\label{thm:reversibility-classification}
Assume $u_i>0$ for all $i$.
\begin{enumerate}
\item If $m\ge 3$, $n\ge 3$, and the fitness vector $(S_1,\dots,S_m)$ is not constant, then Scheme~I is generally nonreversible.
\item If $m\ge 3$, $n\ge 3$, and the fitness vector $(S_1,\dots,S_m)$ is not constant, then Scheme~II is generally nonreversible.
\item Scheme~III is reversible for every $m\ge 2$ and has the explicit stationary distribution derived in Theorem~\ref{thm:scheme-III-multi} below.
\end{enumerate}
\end{theorem}

The finite-state facts used here, including existence of the stationary distribution and the Markov chain tree representation for nonreversible cases, are summarized in Appendix~\ref{app:finite-state}. The proofs of the nonreversibility statements are given in Appendix~\ref{app:nonreversibility}. They use Kolmogorov's cycle criterion: for Schemes~I and~II, a short cycle involving three alleles has unequal clockwise and counterclockwise products whenever two relevant fitnesses differ. Thus, for $m\ge 3$, there is no global detailed-balance product-form stationary distribution for Schemes~I and~II in general.

For $m=2$, however, the state is one-dimensional. Writing $k$ for the number of type-$1$ individuals, the chain moves only between neighboring states $k$ and $k\pm 1$. Hence it is a finite birth--death chain, and detailed balance gives the stationary distribution exactly. This is why Sections~\ref{sec:scheme-I-two} and~\ref{sec:scheme-II-two} focus on two alleles. Scheme~III is treated separately in Section~\ref{sec:scheme-III}, because it remains reversible for any number of alleles. The same one-dimensional observation also makes it possible to combine the three selection mechanisms simultaneously in the two-allele case; this combined model is derived in Section~\ref{sec:combined-two}.

\section{Exact stationary law for two alleles: Scheme I}\label{sec:scheme-I-two}
Set $m=2$ and write the state as $k\in\{0,1,\dots,n\}$, where
\[
 k_1=k,
 \qquad
 k_2=n-k.
\]
Thus a transition $k\to k+1$ replaces a type-$2$ individual by a type-$1$ individual, while a transition $k\to k-1$ replaces a type-$1$ individual by a type-$2$ individual.

For Scheme~I define
\begin{equation}\label{eq:s12-s21}
 s_{12}:=\frac{2S_1}{S_1+S_2},
 \qquad
 s_{21}:=\frac{2S_2}{S_1+S_2},
\end{equation}
and
\begin{equation}\label{eq:scheme-I-two-alphas}
 \widetilde\alpha_1:=\frac{\alpha_1}{s_{12}}
 =\frac{2(n-1)u_1}{1-\Lambda}\frac{S_1+S_2}{2S_1},
 \qquad
 \widetilde\alpha_2:=\frac{\alpha_2}{s_{21}}
 =\frac{2(n-1)u_2}{1-\Lambda}\frac{S_1+S_2}{2S_2}.
\end{equation}
Notice that $s_{12}/s_{21}=S_1/S_2$. The exact two-allele stationary law for Scheme~I is given in the following theorem.

\begin{theorem}
\label{thm:scheme-I-two}
Under Scheme~I with $m=2$, the birth and death probabilities are
\begin{equation}\label{eq:scheme-I-lambda}
\lambda_k:=Q^{(I)}(k,k+1)
=\frac{n-k}{n}\,b s_{12}\bigl(k+\widetilde\alpha_1\bigr),
\qquad 0\le k\le n-1,
\end{equation}
\begin{equation}\label{eq:scheme-I-mu}
\mu_k:=Q^{(I)}(k,k-1)
=\frac{k}{n}\,b s_{21}\bigl(n-k+\widetilde\alpha_2\bigr),
\qquad 1\le k\le n.
\end{equation}
The unique stationary distribution is
\begin{equation}\label{eq:scheme-I-pi}
\pi^{(I)}_k
=\frac{1}{Z_I}\binom{n}{k}(\widetilde\alpha_1)_k(\widetilde\alpha_2)_{n-k}
\left(\frac{s_{12}}{s_{21}}\right)^k,
\qquad 0\le k\le n,
\end{equation}
where $Z_I$ normalizes the probabilities. A compact expression for $Z_I$ is given in Appendix~\ref{app:normalizers}. Equivalently,
\begin{equation}\label{eq:scheme-I-pi-S}
\pi^{(I)}_k
=\frac{1}{Z_I}\binom{n}{k}(\widetilde\alpha_1)_k(\widetilde\alpha_2)_{n-k}
\left(\frac{S_1}{S_2}\right)^k.
\end{equation}
\end{theorem}

\begin{proof}
For a birth--death chain, detailed balance has the one-step form
\begin{equation}\label{eq:bd-detailed-balance-I}
\pi^{(I)}_k\lambda_k=\pi^{(I)}_{k+1}\mu_{k+1},
\qquad 0\le k\le n-1.
\end{equation}
Using Proposition~\ref{prop:scheme-I-kernel},
\[
\lambda_k=\frac{n-k}{n}\left(u_1+b s_{12}k\right)
=\frac{n-k}{n}\,b s_{12}(k+\widetilde\alpha_1),
\]
and
\[
\mu_{k+1}=\frac{k+1}{n}\left(u_2+b s_{21}(n-k-1)\right)
=\frac{k+1}{n}\,b s_{21}(n-k-1+\widetilde\alpha_2).
\]
Therefore
\begin{equation}\label{eq:scheme-I-ratio}
\frac{\pi^{(I)}_{k+1}}{\pi^{(I)}_k}
=\frac{\lambda_k}{\mu_{k+1}}
=\frac{n-k}{k+1}\,\frac{k+\widetilde\alpha_1}{n-k-1+\widetilde\alpha_2}\,\frac{s_{12}}{s_{21}}.
\end{equation}
Iterating \eqref{eq:scheme-I-ratio} from $0$ to $k-1$ gives
\[
\pi^{(I)}_k
=\pi^{(I)}_0
\prod_{r=0}^{k-1}
\frac{n-r}{r+1}\,\frac{r+\widetilde\alpha_1}{n-r-1+\widetilde\alpha_2}\,\frac{s_{12}}{s_{21}}.
\]
The first product is $\binom{n}{k}$. The second product is
\[
\frac{(\widetilde\alpha_1)_k(\widetilde\alpha_2)_{n-k}}{(\widetilde\alpha_2)_n}.
\]
Absorbing the constant $\pi^{(I)}_0/(\widetilde\alpha_2)_n$ into the normalizer yields \eqref{eq:scheme-I-pi}. Since $s_{12}/s_{21}=S_1/S_2$, \eqref{eq:scheme-I-pi-S} follows.
\end{proof}

\begin{remark}
When $S_1=S_2$, one has $s_{12}=s_{21}=1$ and $\widetilde\alpha_i=\alpha_i$. Formula \eqref{eq:scheme-I-pi} then reduces to the neutral beta-binomial law
\[
\pi^{(0)}_k=\frac{1}{Z_0}\binom{n}{k}(\alpha_1)_k(\alpha_2)_{n-k}.
\]
When $S_1\neq S_2$, Scheme~I changes both the effective mutation parameters and the multiplicative fitness tilt.
\end{remark}

\section{Exact stationary law for two alleles: Scheme II}\label{sec:scheme-II-two}
For Scheme~II with two alleles, set
\[
\Delta:=S_1-S_2.
\]
When the state is $k$, the fitness-weighted population total is
\[
\xi(k)=kS_1+(n-k)S_2=nS_2+k\Delta.
\]
The key simplification in the two-allele case is that the state-dependent denominators in the birth transition \(k\to k+1\) and the corresponding death transition \(k+1\to k\) coincide. Indeed,
using the denominator \(\xi(k)-S_j\) from \eqref{eq:TII-final}, the birth
transition \(k\to k+1\) has denominator
\[
\xi(k)-S_2=(n-1)S_2+k\Delta,
\]
while the reverse transition \(k+1\to k\) has denominator
\[
\xi(k+1)-S_1=(n-1)S_2+k\Delta.
\]
This observation motivates the detailed-balance calculation below.

Define
\begin{equation}\label{eq:scheme-II-ABCD}
A_1:=(n-1)u_1S_2,
\qquad
B_1:=u_1\Delta+\frac{1-\Lambda}{2}S_1,
\end{equation}
\begin{equation}\label{eq:scheme-II-CD}
C_2:=(n-1)S_2\left(u_2+\frac{1-\Lambda}{2}\right),
\qquad
D_2:=\frac{1-\Lambda}{2}S_2-u_2\Delta.
\end{equation}
For Scheme~II, the two-allele process admits the following exact stationary distribution.

\begin{theorem}
\label{thm:scheme-II-two}
Under Scheme~II with $m=2$, the birth and death probabilities are
\begin{equation}\label{eq:scheme-II-lambda}
\lambda_k:=Q^{(II)}(k,k+1)
=\frac{n-k}{n}\left(u_1+\frac{(1-\Lambda)S_1k}{2\bigl((n-1)S_2+k\Delta\bigr)}\right),
\qquad 0\le k\le n-1,
\end{equation}
\begin{equation}\label{eq:scheme-II-mu}
\mu_k:=Q^{(II)}(k,k-1)
=\frac{k}{n}\left(u_2+\frac{(1-\Lambda)S_2(n-k)}{2\bigl((n-1)S_2+(k-1)\Delta\bigr)}\right),
\qquad 1\le k\le n.
\end{equation}
The unique stationary distribution is
\begin{equation}\label{eq:scheme-II-product}
\pi^{(II)}_k
=\frac{1}{Z_{II}}\binom{n}{k}
\prod_{r=0}^{k-1}\frac{A_1+B_1r}{C_2-D_2r},
\qquad 0\le k\le n,
\end{equation}
where the empty product, corresponding to \(k=0\), is one and \(Z_{II}\) normalizes the probabilities. A compact expression for \(Z_{II}\) is obtained in Appendix~\ref{app:normalizers}.

When \(B_1\) and \(D_2\) are nonzero, the two affine factors in
\eqref{eq:scheme-II-product} can be written in rising-factorial form, so that
the stationary distribution takes the same beta-binomial-like structure as the
Scheme~I law:
\begin{equation}\label{eq:scheme-II-beta-form}
\pi^{(II)}_k
=\frac{1}{Z_{II}}\binom{n}{k}\theta_{II}^{\,k}(\beta_1)_k(\beta_2)_{n-k}.
\end{equation}
where
\begin{equation}\label{eq:scheme-II-beta-params}
\theta_{II}:=\frac{B_1}{D_2},
\qquad
\beta_1:=\frac{A_1}{B_1}=\frac{(n-1)u_1S_2}{B_1},
\qquad
\beta_2:=\frac{(n-1)u_2S_1}{D_2}.
\end{equation}
A compact hypergeometric expression for the normalizing constant in
\eqref{eq:scheme-II-beta-form} is given in
Appendix~\ref{app:normalizers}.
\end{theorem}

\begin{proof}
By Proposition~\ref{prop:scheme-II-kernel},
\[
\lambda_k
=\frac{n-k}{n}\left(u_1+\frac{(1-\Lambda)S_1k}{2\bigl((n-1)S_2+k\Delta\bigr)}\right).
\]
Putting the terms in the bracket over a common denominator gives
\begin{equation}\label{eq:lambda-affine-II}
\lambda_k
=\frac{n-k}{n}\,\frac{A_1+B_1k}{(n-1)S_2+k\Delta}.
\end{equation}
Similarly,
\begin{align*}
\mu_{k+1}
&=\frac{k+1}{n}\left(u_2+\frac{(1-\Lambda)S_2(n-k-1)}{2\bigl((n-1)S_2+k\Delta\bigr)}\right)\\
&=\frac{k+1}{n}\,\frac{C_2-D_2k}{(n-1)S_2+k\Delta}.
\end{align*}
The common denominator cancels in the detailed-balance ratio, so
\begin{equation}\label{eq:scheme-II-ratio}
\frac{\pi^{(II)}_{k+1}}{\pi^{(II)}_k}
=\frac{\lambda_k}{\mu_{k+1}}
=\frac{n-k}{k+1}\frac{A_1+B_1k}{C_2-D_2k}.
\end{equation}
Iterating \eqref{eq:scheme-II-ratio} gives \eqref{eq:scheme-II-product}.

If \(B_1\) and \(D_2\) are nonzero, then, using the definitions of
\(\beta_1\) and \(\beta_2\), the numerator and denominator in the product
in \eqref{eq:scheme-II-product} can be written as
\[
A_1+B_1r=B_1(r+\beta_1),
\]
and
\[
C_2-D_2r=D_2(\beta_2+n-1-r).
\]
Therefore
\[
\prod_{r=0}^{k-1}\frac{A_1+B_1r}{C_2-D_2r}
=\left(\frac{B_1}{D_2}\right)^k
\frac{(\beta_1)_k(\beta_2)_{n-k}}{(\beta_2)_n}.
\]
Plugging this identity into \eqref{eq:scheme-II-product} and absorbing
\((\beta_2)_n^{-1}\) into the normalizing constant yields
\eqref{eq:scheme-II-beta-form}. This puts the Scheme~II stationary law in the same beta-binomial-like form as the Scheme~I law in
\eqref{eq:scheme-I-pi}.
\end{proof}

\begin{remark}
In the neutral case \(S_1=S_2\), one has \(\Delta=0\),
\(\theta_{II}=1\), and \(\beta_i=\alpha_i\). Thus
\eqref{eq:scheme-II-beta-form} reduces to the neutral beta-binomial law. Away from neutrality, the primary exact expression is the product formula \eqref{eq:scheme-II-product}. The beta-binomial-like expression \eqref{eq:scheme-II-beta-form} is a convenient reparametrization when \(B_1\) and \(D_2\) are nonzero. The excluded cases are the degenerate parameter values for which \(B_1=0\) or \(D_2=0\), so the rising-factorial reparametrization above is not available in this form. In the usual small mutation, moderate selection regime, these degenerate cases are not expected to arise.
\end{remark}

\section{Selection at death: a reversible multi-allele scheme}\label{sec:scheme-III}
Scheme~III differs from Schemes~I and~II because the fitness bias enters through the type chosen for death. As anticipated by
Theorem~\ref{thm:reversibility-classification}, this placement of selection leads to a reversible multi-allele chain with an explicit stationary distribution, as shown in the following theorem.

\begin{theorem}
\label{thm:scheme-III-multi}
For Scheme~III, the multi-allele chain is reversible. Its stationary distribution is
\begin{equation}\label{eq:scheme-III-multi-pi}
\pi^{(III)}(k)
=\frac{1}{Z_{III}}\,
\nu(k)\prod_{\ell=1}^m
\frac{(\alpha_\ell)_{k_\ell}}{k_\ell!}S_\ell^{k_\ell},
\qquad k\in\state,
\end{equation}
where
\[
\nu(k)=\sum_{\ell=1}^m S_\ell^{-1}k_\ell,
\qquad
\alpha_\ell=\frac{2(n-1)u_\ell}{1-\Lambda},
\]
and $Z_{III}$ is the normalizing constant. In the multi-allele case, this normalizer is a finite sum over the state space and can be evaluated directly from \eqref{eq:scheme-III-multi-pi}; no comparably simple one-dimensional hypergeometric form is used here.
\end{theorem}

\begin{proof}
It is enough to check detailed balance for neighboring states. Let
\[
k'=k+e_i-e_j,
\qquad i\neq j,
\]
with $k_j\ge 1$. By Proposition~\ref{prop:scheme-III-kernel}, the forward transition
probability is
\begin{equation}\label{eq:QIII-forward}
Q^{(III)}(k,k')
=\frac{b}{\nu(k)}\frac{k_j}{S_j}(k_i+\alpha_i).
\end{equation}
The corresponding reverse transition probability is
\begin{equation}\label{eq:QIII-reverse}
Q^{(III)}(k',k)
=\frac{b}{\nu(k')}\frac{k_i+1}{S_i}(k_j-1+\alpha_j).
\end{equation}
Thus the detailed-balance ratio required between adjacent states is
\begin{equation}\label{eq:QIII-transition-ratio}
\frac{Q^{(III)}(k,k')}{Q^{(III)}(k',k)}
=\frac{\nu(k')}{\nu(k)}\frac{k_j}{k_i+1}\frac{S_i}{S_j}\frac{k_i+\alpha_i}{k_j-1+\alpha_j}.
\end{equation}

Now define the unnormalized weight
\[
w(k):=\nu(k)\prod_{\ell=1}^m
\frac{(\alpha_\ell)_{k_\ell}}{k_\ell!}S_\ell^{k_\ell}.
\]
The factor \(\nu(k)\) in this weight compensates for the state-dependent
normalization in the Scheme~III transition probabilities
\eqref{eq:QIII-forward}--\eqref{eq:QIII-reverse}.
Then
\begin{align*}
\frac{w(k')}{w(k)}
&=\frac{\nu(k')}{\nu(k)}
\frac{(\alpha_i)_{k_i+1}}{(\alpha_i)_{k_i}}
\frac{k_i!}{(k_i+1)!}
\frac{(\alpha_j)_{k_j-1}}{(\alpha_j)_{k_j}}
\frac{k_j!}{(k_j-1)!}
\frac{S_i}{S_j}\\
&=\frac{\nu(k')}{\nu(k)}
\frac{k_i+\alpha_i}{k_i+1}
\frac{k_j}{k_j-1+\alpha_j}
\frac{S_i}{S_j}.
\end{align*}
This is exactly the right-hand side of \eqref{eq:QIII-transition-ratio}. Hence
\[
w(k)Q^{(III)}(k,k')=w(k')Q^{(III)}(k',k)
\]
for every neighboring pair \(k,k'\). Thus detailed balance holds for the
unnormalized weights \(w\). After normalizing \(w\), we obtain
\eqref{eq:scheme-III-multi-pi}.
\end{proof}

\begin{remark}
In the neutral case \(S_1=\cdots=S_m=1\), the factor \(\nu(k)=n\) is absorbed into the normalizing constant, and \eqref{eq:scheme-III-multi-pi} reduces to the neutral Dirichlet--multinomial stationary distribution. This recovers the neutral multi-allelic Moran process, equivalently the multi-zealot voter model, studied in \citet{braha2026multi}.
\end{remark}

\subsection{The two-allele case of Scheme III}
\label{subsec:scheme-III-two}
The two-allele form follows immediately from the general stationary
distribution in Theorem~\ref{thm:scheme-III-multi}. We present it explicitly to put Scheme~III in the same birth--death notation used for Schemes~I and~II and to prepare for the combined two-allele model in Section~\ref{sec:combined-two}.

For $m=2$, write $k=k_1$ and $n-k=k_2$. Then
\begin{equation}\label{eq:nu-k-two}
\nu_k:=\nu(k,n-k)=\frac{k}{S_1}+\frac{n-k}{S_2}.
\end{equation}
The birth and death probabilities are
\begin{equation}\label{eq:scheme-III-lambda}
\lambda_k:=Q^{(III)}(k,k+1)
=\frac{b}{\nu_k}\frac{n-k}{S_2}(k+\alpha_1),
\qquad 0\le k\le n-1,
\end{equation}
\begin{equation}\label{eq:scheme-III-mu}
\mu_k:=Q^{(III)}(k,k-1)
=\frac{b}{\nu_k}\frac{k}{S_1}(n-k+\alpha_2),
\qquad 1\le k\le n.
\end{equation}
Thus
\begin{equation}\label{eq:scheme-III-two-pi}
\pi^{(III)}_k
=\frac{1}{Z_{III}}\,
\nu_k\binom{n}{k}(\alpha_1)_k(\alpha_2)_{n-k}S_1^kS_2^{n-k},
\qquad 0\le k\le n,
\end{equation}
where \(Z_{III}\) denotes the normalizing constant for this two-allele stationary distribution. Thus \eqref{eq:scheme-III-two-pi} is the two-allele form of \eqref{eq:scheme-III-multi-pi}, with the constant factor \(n!\) absorbed into the normalizer. Appendix~\ref{app:normalizers} gives the corresponding two-allele normalizing constant.

\begin{proof}
For completeness, we give the direct proof in the two-allele case. 

The transition probabilities \eqref{eq:scheme-III-lambda} and \eqref{eq:scheme-III-mu} follow directly from \eqref{eq:QIII-alpha-form}. Applying one-step detailed balance gives
\[
\frac{\pi^{(III)}_{k+1}}{\pi^{(III)}_k}
=\frac{\lambda_k}{\mu_{k+1}}
=\frac{\nu_{k+1}}{\nu_k}\frac{n-k}{k+1}\frac{S_1}{S_2}\frac{k+\alpha_1}{n-k-1+\alpha_2}.
\]
Iterating this ratio from $0$ to $k-1$ yields
\[
\pi^{(III)}_k
\propto
\frac{\nu_k}{\nu_0}
\binom{n}{k}
\frac{(\alpha_1)_k(\alpha_2)_{n-k}}{(\alpha_2)_n}
\left(\frac{S_1}{S_2}\right)^k.
\]
Multiplying by the constant factors $\nu_0(\alpha_2)_nS_2^n$ and absorbing them into the normalizer gives \eqref{eq:scheme-III-two-pi}.
\end{proof}

\section{Combining all three selection mechanisms for two alleles}
\label{sec:combined-two}

The previous sections isolate three different biological stages at which selection can enter the Moran update: during offspring copying, during mate choice, or during death. In many applications these mechanisms need not be mutually exclusive. A fitter allele may make its carrier less likely to be removed, more likely to be chosen as a mate, and more likely to be copied once two parents have been selected. This section shows that, in the two-allele case, all three effects can be incorporated simultaneously without losing an exact
stationary distribution. For three or more alleles, however, the combined process should generally be treated as a multidimensional nonreversible chain unless additional symmetry assumptions are imposed.

The reason that the two-allele case remains exactly solvable is structural. With two alleles, the state is still one-dimensional, so the process is a finite birth--death chain. Thus, even though the combined mechanism is more complex than any single scheme alone, one-step detailed balance still determines the stationary law.

The combined update is defined as follows. First, the individual to be replaced is chosen with probability proportional to inverse fitness, as in Scheme~III. Second, conditional on the type of the individual chosen for replacement, the second parent is chosen from the remaining population with probability proportional to fitness, as in Scheme~II. Third, after the two parents have been selected, the offspring copies one of the two parental alleles with probability proportional to parental fitness, as in Scheme~I. Parent-independent mutation then acts on the copied allele.

Set
\[
 k_1=k,
 \qquad
 k_2=n-k,
\]
and recall the notation
\begin{equation}\label{eq:combined-Delta-nu}
\Delta:=S_1-S_2,
\qquad
\nu_k:=\frac{k}{S_1}+\frac{n-k}{S_2}.
\end{equation}
Let $Q^{(C)}$ denote the resulting transition matrix. For compact notation, define
\begin{equation}\label{eq:combined-eta-Dk}
\eta:=\frac{1-\Lambda}{S_1+S_2},
\qquad
D_k:=(n-1)S_2+k\Delta .
\end{equation}
Here $D_k$ is the total fitness of the available second-parent pool when the state is $k$ and the focal individual chosen for replacement is of type~2.

The following theorem gives the exact stationary distribution for the combined two-allele process.

\begin{theorem}
\label{thm:combined-two}
For the two-allele model in which death selection, fitness-biased mate choice, and fitness-biased offspring copying act simultaneously, the birth probabilities are
\begin{equation}\label{eq:combined-lambda}
\lambda_k
:=Q^{(C)}(k,k+1)
=\frac{n-k}{S_2\nu_k}
\left[
 u_1+\eta\frac{S_1^2k}{D_k}
\right],
\qquad 0\le k\le n-1,
\end{equation}
and the death probabilities are
\begin{equation}\label{eq:combined-mu}
\mu_k
:=Q^{(C)}(k,k-1)
=\frac{k}{S_1\nu_k}
\left[
 u_2+\eta\frac{S_2^2(n-k)}{(n-1)S_2+(k-1)\Delta}
\right],
\qquad 1\le k\le n.
\end{equation}
Define the four coefficients
\begin{equation}\label{eq:combined-affine-constants}
a_{+}:=(n-1)u_1S_2,
\qquad
b_{+}:=u_1\Delta+\eta S_1^2,
\qquad
a_{-}:=(n-1)S_2(u_2+\eta S_2),
\qquad
b_{-}:=\eta S_2^2-u_2\Delta.
\end{equation}
The plus signs refer to the linear factor associated with the birth transition
$k\to k+1$, and the minus signs refer to the linear factor associated with the
reverse transition $k+1\to k$.
Then the unique stationary distribution is
\begin{equation}\label{eq:combined-pi-product}
\pi^{(C)}_k
=\frac{1}{Z_C}\,
\nu_k\binom{n}{k}
\left(\frac{S_1}{S_2}\right)^k
\prod_{r=0}^{k-1}\frac{a_{+}+b_{+}r}{a_{-}-b_{-}r},
\qquad 0\le k\le n,
\end{equation}
where the empty product, corresponding to $k=0$, is one, and
\begin{equation}\label{eq:combined-normalizer}
Z_C=
\sum_{k=0}^n
\nu_k\binom{n}{k}
\left(\frac{S_1}{S_2}\right)^k
\prod_{r=0}^{k-1}\frac{a_{+}+b_{+}r}{a_{-}-b_{-}r}.
\end{equation}
\end{theorem}

\begin{proof}
We first derive the birth probability $\lambda_k$, corresponding to the transition $k\to k+1$. This transition occurs when a type-$2$ individual is selected for replacement and the replacement offspring ultimately has type~1.

Under death selection, the probability of selecting a type-$2$ individual for replacement is
\[
\frac{S_2^{-1}(n-k)}{S_1^{-1}k+S_2^{-1}(n-k)}
=
\frac{n-k}{S_2\nu_k}.
\]
After this type-$2$ focal individual has been selected, the available second-parent pool contains $k$ type-$1$ individuals and $n-k-1$ type-$2$ individuals. Its total fitness is
\[
D_k=kS_1+(n-k-1)S_2=(n-1)S_2+k\Delta .
\]
Thus a type-$1$ second parent is chosen with probability $S_1k/D_k$, and a type-$2$ second parent is chosen with probability $S_2(n-k-1)/D_k$.

Let $T^{(C)}_{1\mid 2}(k)$ be the conditional probability that the offspring is type~1, given that the focal individual selected for replacement is type~2. There are two possible second-parent types.

First, suppose the second parent has type~1. This event has probability $S_1k/D_k$. The offspring then ends as type~1 in two ways: it copies the focal type-$2$ allele and mutates to type~1, or it copies the type-$1$ second parent and remains type~1. Therefore this contribution is
\[
A_{+}
=
\frac{S_1k}{D_k}
\left[
\frac{S_2}{S_1+S_2}u_1
+
\frac{S_1}{S_1+S_2}(1-\Lambda+u_1)
\right].
\]

Second, suppose the second parent has type~2. This event has probability $S_2(n-k-1)/D_k$. Both parental alleles are then type~2, so the offspring must mutate to type~1. The contribution is
\[
B_{+}
=
\frac{S_2(n-k-1)}{D_k}u_1 .
\]

Adding these mutually exclusive contributions gives
\begin{align*}
T^{(C)}_{1\mid 2}(k)
&=A_{+}+B_{+} \\
&=
\frac{S_1k}{D_k}
\left[
\frac{S_2}{S_1+S_2}u_1
+
\frac{S_1}{S_1+S_2}(1-\Lambda+u_1)
\right]
+
\frac{S_2(n-k-1)}{D_k}u_1 \\
&=
\frac{1}{D_k}
\left\{
\left[S_1k+S_2(n-k-1)\right]u_1
+
\frac{1-\Lambda}{S_1+S_2}S_1^2k
\right\}.
\end{align*}
Since $S_1k+S_2(n-k-1)=D_k$, and since $\eta=(1-\Lambda)/(S_1+S_2)$, we obtain
\[
T^{(C)}_{1\mid 2}(k)
=
u_1+\eta\frac{S_1^2k}{D_k}.
\]
Multiplying by the death-selection probability for a type-$2$ focal individual gives
\[
\lambda_k
=
\frac{n-k}{S_2\nu_k}
\left[
u_1+\eta\frac{S_1^2k}{D_k}
\right],
\]
which is \eqref{eq:combined-lambda}.

We next derive the death probability $\mu_k$, corresponding to the transition $k\to k-1$. This transition occurs when a type-$1$ individual is selected for replacement and the replacement offspring ultimately has type~2. The probability of selecting a type-$1$ individual for replacement is
\[
\frac{S_1^{-1}k}{S_1^{-1}k+S_2^{-1}(n-k)}
=
\frac{k}{S_1\nu_k}.
\]
After this type-$1$ focal individual has been selected, the available second-parent pool contains $k-1$ type-$1$ individuals and $n-k$ type-$2$ individuals. Its total fitness is
\[
E_k:=(k-1)S_1+(n-k)S_2=(n-1)S_2+(k-1)\Delta .
\]

Let $T^{(C)}_{2\mid 1}(k)$ be the conditional probability that the offspring is type~2, given that the focal individual selected for replacement is type~1. Again there are two possible second-parent types.

First, suppose the second parent has type~2. This event has probability $S_2(n-k)/E_k$. The offspring then ends as type~2 in two ways: it copies the focal type-$1$ allele and mutates to type~2, or it copies the type-$2$ second parent and remains type~2. The contribution is
\[
A_{-}
=
\frac{S_2(n-k)}{E_k}
\left[
\frac{S_1}{S_1+S_2}u_2
+
\frac{S_2}{S_1+S_2}(1-\Lambda+u_2)
\right].
\]

Second, suppose the second parent has type~1. This event has probability $S_1(k-1)/E_k$. Both parental alleles are then type~1, so the offspring must mutate to type~2. The contribution is
\[
B_{-}
=
\frac{S_1(k-1)}{E_k}u_2 .
\]

Adding these contributions gives
\begin{align*}
T^{(C)}_{2\mid 1}(k)
&=A_{-}+B_{-} \\
&=
\frac{S_2(n-k)}{E_k}
\left[
\frac{S_1}{S_1+S_2}u_2
+
\frac{S_2}{S_1+S_2}(1-\Lambda+u_2)
\right]
+
\frac{S_1(k-1)}{E_k}u_2 \\
&=
\frac{1}{E_k}
\left\{
\left[S_2(n-k)+S_1(k-1)\right]u_2
+
\frac{1-\Lambda}{S_1+S_2}S_2^2(n-k)
\right\}.
\end{align*}
Since $S_2(n-k)+S_1(k-1)=E_k$, this becomes
\[
T^{(C)}_{2\mid 1}(k)
=
u_2+\eta\frac{S_2^2(n-k)}{E_k}.
\]
Multiplying by the death-selection probability for a type-$1$ focal individual gives
\[
\mu_k
=
\frac{k}{S_1\nu_k}
\left[
u_2+\eta\frac{S_2^2(n-k)}{(n-1)S_2+(k-1)\Delta}
\right],
\]
which is \eqref{eq:combined-mu}.

We now put the two transition probabilities into adjacent-ratio form. From \eqref{eq:combined-lambda},
\begin{align*}
\lambda_k
&=\frac{n-k}{S_2\nu_k}
\left[u_1+\eta\frac{S_1^2k}{D_k}\right]
=\frac{n-k}{S_2\nu_k}\frac{u_1D_k+\eta S_1^2k}{D_k} \\
&=\frac{n-k}{S_2\nu_k}\frac{a_{+}+b_{+}k}{D_k},
\end{align*}
where $a_{+}=(n-1)u_1S_2$ and $b_{+}=u_1\Delta+\eta S_1^2$. Similarly, evaluating \eqref{eq:combined-mu} at $k+1$ gives
\begin{align*}
\mu_{k+1}
&=\frac{k+1}{S_1\nu_{k+1}}
\left[u_2+\eta\frac{S_2^2(n-k-1)}{D_k}\right] \\
&=\frac{k+1}{S_1\nu_{k+1}}\frac{u_2D_k+\eta S_2^2(n-k-1)}{D_k} \\
&=\frac{k+1}{S_1\nu_{k+1}}\frac{a_{-}-b_{-}k}{D_k},
\end{align*}
where $a_{-}=(n-1)S_2(u_2+\eta S_2)$ and $b_{-}=\eta S_2^2-u_2\Delta$. Therefore
\[
\frac{\lambda_k}{\mu_{k+1}}
=
\frac{n-k}{k+1}\frac{S_1}{S_2}\frac{\nu_{k+1}}{\nu_k}
\frac{a_{+}+b_{+}k}{a_{-}-b_{-}k}.
\]

Because the two-allele chain is a finite birth--death chain with positive mutation, one-step detailed balance determines the stationary distribution. Iterating
\[
\frac{\pi^{(C)}_{k+1}}{\pi^{(C)}_k}=\frac{\lambda_k}{\mu_{k+1}}
\]
from $0$ to $k-1$ gives
\[
\pi^{(C)}_k
\propto
\nu_k\binom{n}{k}
\left(\frac{S_1}{S_2}\right)^k
\prod_{r=0}^{k-1}\frac{a_{+}+b_{+}r}{a_{-}-b_{-}r}.
\]
Normalizing over $k=0,1,\ldots,n$ gives \eqref{eq:combined-pi-product} and \eqref{eq:combined-normalizer}.
\end{proof}

\begin{remark}
The product form \eqref{eq:combined-pi-product} is the primary expression, because it remains valid in all parameter regimes. When $b_{+}b_{-}\neq 0$, it can be written more compactly in rising-factorial notation. Define
\begin{equation}\label{eq:combined-beta-params}
\theta_C:=\frac{S_1}{S_2}\frac{b_{+}}{b_{-}},
\qquad
\beta_{+}:=\frac{a_{+}}{b_{+}},
\qquad
\beta_{-}:=\frac{(n-1)u_2S_1}{b_{-}}.
\end{equation}
Then
\begin{equation}\label{eq:combined-pi-beta}
\pi^{(C)}_k
=\frac{1}{\widetilde Z_C}\,
\nu_k\binom{n}{k}
\theta_C^{\,k}(\beta_{+})_k(\beta_{-})_{n-k},
\qquad 0\le k\le n.
\end{equation}
This representation makes the analogy with the two-allele formulas for Schemes~I--III transparent: the combined mechanism produces a
beta-binomial-type factor, modified by the death-selection factor $\nu_k$ and by effective parameters that incorporate all three selection mechanisms.
\end{remark}

\begin{remark}
If $S_1=S_2$, the combined mechanism reduces to the neutral two-parent Moran update. In this case $\nu_k$ is constant, $\theta_C=1$, $\beta_{+}=\alpha_1$, and $\beta_{-}=\alpha_2$, so \eqref{eq:combined-pi-beta} reduces to the neutral beta-binomial stationary law. Away from neutrality, the stationary distribution
separates the effects of the three selection mechanisms: the factor $\nu_k$ comes from fitness-biased death, the factor $(S_1/S_2)^k$ comes from fitness-biased copying between the two parents, and the product of linear factors captures how fitness-biased mate choice changes the effective reproduction terms. Thus the combined model remains exactly solvable for two alleles, but its stationary distribution contains identifiable factors from all three selection mechanisms.
\end{remark}

\section{Weak selection: one slightly fitter allele}\label{sec:weak-selection}

This section derives the first-order weak selection forms of the stationary distributions for the three schemes. The point is to show that, when selection is small, the exact stationary distributions are perturbations of the neutral benchmark: in the two-allele case, the neutral beta-binomial distribution, and for the multi-allele form of Scheme~III, the neutral Dirichlet--multinomial law recalled in Appendix~\ref{app:neutral-DM}.

Assume that allele~1 is slightly fitter than the other allele or alleles:
\[
S_1=1+\varepsilon,
\qquad
S_i=1\quad (i\ge 2),
\qquad |\varepsilon|\ll 1.
\]
The mutation parameters $u_i$, and hence the parameters $\alpha_i$ defined in \eqref{eq:b-alpha-def}, are held fixed as $\varepsilon$ varies. For Schemes~I and~II we use the exact two-allele laws, because for three or more alleles these schemes are generally nonreversible. Scheme~III is reversible for any number of alleles, so we give its weak selection expansion in both the two-allele and multi-allele cases.

The rest of this introductory part sets up the notation and derives a weak selection perturbation formula that will be used in the scheme-specific analyses below. More specifically, the scheme-specific calculations all have the same structure. First, the exact stationary weight is expanded as the neutral weight multiplied by a first-order correction. This unnormalized correction must then be converted into a correction of the normalized stationary
probabilities. The formula below gives this conversion in a form that can be applied to each scheme: the first-order correction is centered by subtracting its neutral expectation.

We first introduce the finite harmonic sum
\begin{equation}\label{eq:H-a-r-def}
H_r(a):=\sum_{q=0}^{r-1}\frac{1}{a+q},
\qquad H_0(a):=0.
\end{equation}
The following first-order Taylor expansion of the rising factorial will be used repeatedly. For fixed $a>0$, fixed integer $r\ge 0$, and fixed constant $c$,
\begin{equation}\label{eq:rising-weak-expansion}
\bigl(a(1+c\varepsilon)\bigr)_r
=(a)_r\left[1+c\varepsilon\,aH_r(a)+O(\varepsilon^2)\right].
\end{equation}

For the two-allele formulas, define the unnormalized neutral weight by
\begin{equation}\label{eq:neutral-two-weight}
W_0(k):=\binom{n}{k}(\alpha_1)_k(\alpha_2)_{n-k},
\qquad 0\le k\le n.
\end{equation}
The corresponding neutral stationary distribution is the normalized
beta-binomial law
\begin{equation}\label{eq:weak-neutral-pi0}
\pi^{(0)}_k:=\mathbb P_0(K=k)
=\frac{W_0(k)}{Z_0}
=\frac{\binom{n}{k}(\alpha_1)_k(\alpha_2)_{n-k}}{(\alpha_0)_n},
\end{equation}
where $\alpha_0:=\alpha_1+\alpha_2$, and
\[
Z_0=\sum_{k=0}^n W_0(k)=(\alpha_0)_n .
\]
by the Vandermonde identity. Here $K$ denotes a random variable with distribution $\pi^{(0)}$. The subscript $0$ indicates quantities associated with the neutral beta-binomial stationary distribution; in particular, $\mathbb P_0$ and $\mathbb E_0$ denote probability
and expectation under this neutral law. Thus, for any function
$f:\{0,1,\ldots,n\}\to\mathbb R$,
\begin{equation}\label{eq:weak-E0-def}
\mathbb E_0[f(K)]
:=\sum_{k=0}^n f(k)\pi^{(0)}_k.
\end{equation}
In particular, the neutral mean is
\begin{equation}\label{eq:weak-neutral-mean}
\mu_0:=\mathbb E_0[K]=\frac{n\alpha_1}{\alpha_0}.
\end{equation}

We next derive a simple weak-selection perturbation formula that will be used repeatedly in the scheme-specific expansions below. Suppose that, for each state $k$, an unnormalized weak selection weight can be written as
\begin{equation}\label{eq:generic-weak-weight}
W_\varepsilon(k)=W_0(k)\left[1+\varepsilon h(k)+O(\varepsilon^2)\right].
\end{equation}
Here $W_0(k)$ is the unnormalized neutral weight in
\eqref{eq:neutral-two-weight}, and $h(k)$ is the first-order change relative to that neutral unnormalized weight. Because there are only finitely many states $k=0,1,\ldots,n$, the expansion may be summed term by term. Hence
\begin{align}
Z_\varepsilon
&:=\sum_{k=0}^n W_\varepsilon(k) \notag\\
&=\sum_{k=0}^n W_0(k)
+\varepsilon\sum_{k=0}^n W_0(k)h(k)
+O(\varepsilon^2) \notag\\
&=Z_0\left[1+
\varepsilon\sum_{k=0}^n h(k)\frac{W_0(k)}{Z_0}
+O(\varepsilon^2)\right] \notag\\
&=Z_0\left[1+\varepsilon\mathbb E_0[h(K)]+O(\varepsilon^2)\right].
\label{eq:weak-normalizer-expansion}
\end{align}
Thus $\mathbb E_0[h(K)]$ is the neutral expectation of the first-order correction $h(k)$ appearing in \eqref{eq:generic-weak-weight}, and it is the term that enters the expansion of the normalizing constant in \eqref{eq:weak-normalizer-expansion}.

The normalized probability, after incorporating weak selection, is
\[
\pi_\varepsilon(k)=\frac{W_\varepsilon(k)}{Z_\varepsilon}.
\]
Using \eqref{eq:generic-weak-weight} and
\eqref{eq:weak-normalizer-expansion}, we obtain
\begin{align}
\pi_\varepsilon(k)
&=\frac{W_0(k)}{Z_0}
\frac{1+\varepsilon h(k)+O(\varepsilon^2)}
{1+\varepsilon\mathbb E_0[h(K)]+O(\varepsilon^2)} \notag\\
&=\pi^{(0)}_k
\left[1+\varepsilon h(k)+O(\varepsilon^2)\right]
\left[1-\varepsilon\mathbb E_0[h(K)]+O(\varepsilon^2)\right] \notag\\
&=\pi^{(0)}_k
\left[1+\varepsilon\left(h(k)-\mathbb E_0[h(K)]\right)
+O(\varepsilon^2)\right].
\label{eq:centered-weak-form}
\end{align}
Equation~\eqref{eq:centered-weak-form} can be interpreted as follows: the uncentered function $h(k)$ describes the first-order tilt of the unnormalized weights, as shown in \eqref{eq:generic-weak-weight}, while the centered quantity $h(k)-\mathbb E_0[h(K)]$ gives the first-order change in the normalized stationary probabilities.

\subsection{Scheme I under weak selection}
For Scheme~I, with $S_1=1+\varepsilon$ and $S_2=1$,
\[
s_{12}=\frac{2(1+\varepsilon)}{2+\varepsilon}
=1+\frac{\varepsilon}{2}+O(\varepsilon^2),
\qquad
s_{21}=\frac{2}{2+\varepsilon}
=1-\frac{\varepsilon}{2}+O(\varepsilon^2).
\]
Therefore
\[
\widetilde\alpha_1=\frac{\alpha_1}{s_{12}}
=\alpha_1\left(1-\frac{\varepsilon}{2}\right)+O(\varepsilon^2),
\qquad
\widetilde\alpha_2=\frac{\alpha_2}{s_{21}}
=\alpha_2\left(1+\frac{\varepsilon}{2}\right)+O(\varepsilon^2),
\]
and
\[
\left(\frac{S_1}{S_2}\right)^k=(1+\varepsilon)^k
=1+\varepsilon k+O(\varepsilon^2).
\]
Substituting these expansions into \eqref{eq:scheme-I-pi-S} in
Theorem~\ref{thm:scheme-I-two} and using \eqref{eq:rising-weak-expansion} gives
\begin{equation}\label{eq:hI-weak}
h_I(k)
=k-\frac{\alpha_1}{2}H_k(\alpha_1)
+\frac{\alpha_2}{2}H_{n-k}(\alpha_2).
\end{equation}
Applying the weak-selection perturbation formula
\eqref{eq:centered-weak-form} with $h=h_I$ gives
\begin{equation}\label{eq:scheme-I-weak-pi}
\pi^{(I)}_\varepsilon(k)
=\pi^{(0)}_k
\left[1+\varepsilon\left(h_I(k)-\mathbb E_0[h_I(K)]\right)+O(\varepsilon^2)\right].
\end{equation}
The centering constant in \eqref{eq:scheme-I-weak-pi} has a simple explicit form:
\begin{equation}\label{eq:hI-weak-mean}
\mathbb E_0[h_I(K)]
=\mu_0-\frac{\alpha_1-\alpha_2}{2}H_n(\alpha_0).
\end{equation}
To see this, first note that the neutral beta-binomial law satisfies
\begin{equation}\label{eq:neutral-bb-expectations}
\mathbb E_0[K]=\mu_0,
\qquad
\mathbb E_0[H_K(\alpha_1)]=H_n(\alpha_0),
\qquad
\mathbb E_0[H_{n-K}(\alpha_2)]=H_n(\alpha_0).
\end{equation}
These beta-binomial identities are derived in
Appendix~\ref{app:weak-beta-binomial-identities}. Substituting them into
\eqref{eq:hI-weak} gives \eqref{eq:hI-weak-mean}.

The uncentered correction $h_I(k)$ in \eqref{eq:hI-weak} has three
contributions, corresponding to the three first-order expansions above. The term $k$ comes from the explicit fitness factor
\[
\left(\frac{S_1}{S_2}\right)^k=(1+\varepsilon)^k
=1+\varepsilon k+O(\varepsilon^2).
\]
The term $-(\alpha_1/2)H_k(\alpha_1)$ comes from the first-order change in $(\widetilde\alpha_1)_k$, since
$\widetilde\alpha_1=\alpha_1(1-\varepsilon/2)+O(\varepsilon^2)$. The term
$(\alpha_2/2)H_{n-k}(\alpha_2)$ comes from the first-order change in
$(\widetilde\alpha_2)_{n-k}$, since
$\widetilde\alpha_2=\alpha_2(1+\varepsilon/2)+O(\varepsilon^2)$. Thus Scheme~I affects the stationary weights in two related ways: it introduces the explicit fitness factor $(S_1/S_2)^k$, and it changes the effective mutation parameters appearing in the two rising factorials.

\subsection{Scheme II under weak selection}
For Scheme~II, the beta-binomial-like representation
\eqref{eq:scheme-II-beta-form} in Theorem~\ref{thm:scheme-II-two} is
especially convenient. When $S_1=1+\varepsilon$ and $S_2=1$, the parameters in \eqref{eq:scheme-II-beta-params} satisfy
\begin{align}
\theta_{II}
&=1+\varepsilon\left(1+\frac{\alpha_1+\alpha_2}{n-1}\right)+O(\varepsilon^2),\label{eq:thetaII-weak}\\
\beta_1
&=\alpha_1\left[1-\varepsilon\left(1+\frac{\alpha_1}{n-1}\right)\right]+O(\varepsilon^2),\label{eq:beta1II-weak}\\
\beta_2
&=\alpha_2\left[1+\varepsilon\left(1+\frac{\alpha_2}{n-1}\right)\right]+O(\varepsilon^2).\label{eq:beta2II-weak}
\end{align}
Substitution into \eqref{eq:scheme-II-beta-form} gives
\begin{equation}\label{eq:hII-weak}
h_{II}(k)
=\left(1+\frac{\alpha_1+\alpha_2}{n-1}\right)k
-\alpha_1\left(1+\frac{\alpha_1}{n-1}\right)H_k(\alpha_1)
+\alpha_2\left(1+\frac{\alpha_2}{n-1}\right)H_{n-k}(\alpha_2).
\end{equation}
Applying the weak-selection perturbation formula
\eqref{eq:centered-weak-form} with $h=h_{II}$ gives
\begin{equation}\label{eq:scheme-II-weak-pi}
\pi^{(II)}_\varepsilon(k)
=\pi^{(0)}_k
\left[1+\varepsilon\left(h_{II}(k)-\mathbb E_0[h_{II}(K)]\right)+O(\varepsilon^2)\right].
\end{equation}
The centering constant is again explicit:
\begin{align}
\mathbb E_0[h_{II}(K)]
&=\left(1+\frac{\alpha_0}{n-1}\right)\mu_0
-\left[\alpha_1\left(1+\frac{\alpha_1}{n-1}\right)
-\alpha_2\left(1+\frac{\alpha_2}{n-1}\right)\right]H_n(\alpha_0) \notag\\
&=\left(1+\frac{\alpha_0}{n-1}\right)
\left[\mu_0-(\alpha_1-\alpha_2)H_n(\alpha_0)\right] \notag\\
&=\frac{1+\Lambda}{1-\Lambda}
\left[\mu_0-(\alpha_1-\alpha_2)H_n(\alpha_0)\right].
\label{eq:hII-weak-mean}
\end{align}
The first equality follows by applying the neutral expectations in
\eqref{eq:neutral-bb-expectations} to \eqref{eq:hII-weak}. The second equality
uses
\[
\alpha_1\left(1+\frac{\alpha_1}{n-1}\right)
-\alpha_2\left(1+\frac{\alpha_2}{n-1}\right)
=(\alpha_1-\alpha_2)\left(1+\frac{\alpha_0}{n-1}\right),
\]
and the third equality uses
\[
1+\frac{\alpha_0}{n-1}=\frac{1+\Lambda}{1-\Lambda}.
\]

The correction $h_{II}(k)$ in \eqref{eq:hII-weak} has the same general structure as the Scheme~I correction, but with different coefficients. As in Scheme~I, there is a term that tilts the weights toward states with more copies of allele~1, together with two harmonic-sum terms coming from the first-order changes in the effective beta-binomial parameters. In Scheme~II, however, these coefficients also contain factors of $1/(n-1)$, because selection acts when the second parent is chosen from the remaining $n-1$
individuals. Thus the weak-selection correction reflects both the fitness advantage of allele~1 and the finite size of the available second-parent pool.

\subsection{Scheme III under weak selection}
Scheme~III has the simplest weak selection expansion: unlike Schemes~I and~II, selection does not perturb the rising-factorial parameters in the neutral beta-binomial or Dirichlet--multinomial core. In the two-allele case, the Scheme~III formula
\eqref{eq:scheme-III-two-pi} from Section~\ref{subsec:scheme-III-two} gives the unnormalized weight
\[
W^{(III)}_\varepsilon(k)
=\nu_k\binom{n}{k}(\alpha_1)_k(\alpha_2)_{n-k}(1+\varepsilon)^k.
\]
where
\[
\nu_k=\frac{k}{1+\varepsilon}+n-k
=n-\varepsilon k+O(\varepsilon^2).
\]
Therefore
\[
\nu_k(1+\varepsilon)^k
=n\left[1+\varepsilon\left(1-\frac1n\right)k+O(\varepsilon^2)\right],
\]
and the constant factor $n$ is absorbed into the normalizer. Thus
\begin{equation}\label{eq:hIII-weak-two}
h_{III}(k)=\left(1-\frac1n\right)k,
\end{equation}
and
\begin{equation}\label{eq:scheme-III-weak-pi-two}
\pi^{(III)}_\varepsilon(k)
=\pi^{(0)}_k
\left[1+\varepsilon\left(1-\frac1n\right)
\left(k-\mathbb E_0[K]\right)+O(\varepsilon^2)\right].
\end{equation}
Since, under the neutral stationary law,
\[
\mathbb E_0[K]=\frac{n\alpha_1}{\alpha_1+\alpha_2},
\]
this can also be written as
\[
\pi^{(III)}_\varepsilon(k)
=\pi^{(0)}_k
\left[1+\varepsilon\left(1-\frac1n\right)
\left(k-\frac{n\alpha_1}{\alpha_1+\alpha_2}\right)+O(\varepsilon^2)\right].
\]

For the multi-allele Scheme~III distribution, let
\[
\pi^{(0)}(k)=\frac{n!}{(\alpha_0)_n}\prod_{i=1}^m\frac{(\alpha_i)_{k_i}}{k_i!},
\qquad
\alpha_0=\sum_{i=1}^m\alpha_i,
\]
be the neutral Dirichlet--multinomial distribution. This is the neutral multi-allelic Moran stationary law, equivalently the multi-zealot voter-model stationary law, studied in \citet{braha2026multi}. If $S_1=1+\varepsilon$ and $S_i=1$ for $i\ge 2$, then
\[
\nu(k)=\frac{k_1}{1+\varepsilon}+\sum_{i=2}^m k_i
=n-\varepsilon k_1+O(\varepsilon^2),
\]
and
\[
\prod_{i=1}^m S_i^{k_i}=(1+\varepsilon)^{k_1}
=1+\varepsilon k_1+O(\varepsilon^2).
\]
Consequently,
\begin{equation}\label{eq:scheme-III-weak-pi-multi}
\pi^{(III)}_\varepsilon(k)
=\pi^{(0)}(k)
\left[1+\varepsilon\left(1-\frac1n\right)
\left(k_1-\frac{n\alpha_1}{\alpha_0}\right)+O(\varepsilon^2)\right].
\end{equation}
Thus, for Scheme~III, weak selection produces a particularly simple
perturbation of the neutral law. It tilts the neutral beta-binomial or Dirichlet--multinomial distribution toward states with more copies of the fitter allele, but the strength of this tilt is reduced by the factor \(1-1/n\). This reduction reflects the death-selection mechanism: although fitter individuals contribute positively through the factor \(S_1^{k_1}\), they are also less likely to be chosen for replacement, and this state-dependent
death normalization partially offsets the direct fitness tilt. Unlike Schemes~I and~II, no harmonic-sum terms appear, because selection does not perturb the rising-factorial mutation terms. In Scheme~III, weak selection therefore acts only through the explicit fitness factor and the state-dependent normalization in the death step, leaving the neutral mutation core unchanged.

\section{Numerical illustrations}
\label{sec:numerical-illustrations}

We illustrate the exact stationary distributions by plotting the probability \(\pi_k\) of finding \(k\) copies of allele~1 in a population of size \(n=100\). The figures compare different locations of selection in the Moran update and different values of the effective mutation parameters \(\alpha_1\) and
\(\alpha_2\), defined in \eqref{eq:b-alpha-def}.

The three symmetric choices \(\alpha_1=\alpha_2=0.5\),
\(\alpha_1=\alpha_2=1\), and \(\alpha_1=\alpha_2=2\) are chosen to represent regimes below, at, and above the neutral critical mutation rate. More precisely, in the neutral model \(S_1=S_2=1\), the stationary distribution of the labeled allele-count vector has Dirichlet--multinomial form and is equivalent to the multi-zealot voter-model stationary law studied in \citet{braha2026multi}.
Under the transformation from mutation rates to effective mutation parameters in \eqref{eq:b-alpha-def}, the neutral critical mutation rate corresponds to \(\alpha_1=\alpha_2=1\). Thus \(\alpha_1=\alpha_2=0.5\) is below the neutral critical value, \(\alpha_1=\alpha_2=1\) is at the neutral critical value, and
\(\alpha_1=\alpha_2=2\) is above the neutral critical value.

Figure~\ref{fig1} compares the stationary distributions generated by the three single mechanisms and their combination for the above-critical case \(\alpha_1=\alpha_2=2\), with \(S_1=1.10\) and \(S_2=1\). The neutral beta-binomial distribution is shown in black. The colored curves correspond to selection during mate choice, selection at death, and selection during reproduction, while the dashed curve shows the combined effect of all three selection mechanisms. Although the three single-mechanism curves are close over much of the
distribution, especially in the lower tail, they separate visibly in the upper tail. In that region, selection at death produces the largest departure from neutrality, followed by selection during reproduction and then selection during mate choice. The combined process produces a much larger effect because the three mechanisms act in the same direction, all favoring states with more copies of the fitter allele. For effective mutation values at or below the neutral critical value, the stage-dependent differences among the single-mechanism curves remain present but are less visually pronounced than in the above-critical case. For this reason, Figure~\ref{fig2} focuses on one representative mechanism while varying the strength of selection.

\begin{figure}[htb!]
    \centering
    \includegraphics[width=0.7\textwidth]{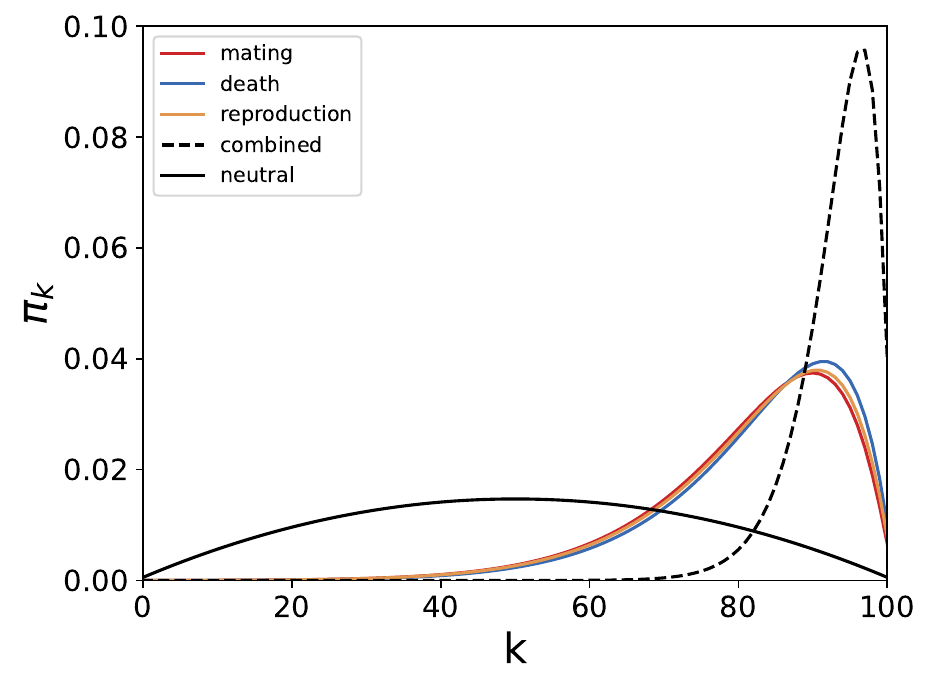}
    \caption{Probability \(\pi_k\) of finding \(k\) copies of allele~1 in a population of size \(n=100\), with effective mutation parameters \(\alpha_1=\alpha_2=2\) above the neutral critical value. The black line shows the neutral case \(S_1=S_2=1\). The red, blue, and orange curves show the distributions for \(S_1=1.10\) and \(S_2=1\) when selection acts during mate choice, at death, and during reproduction, respectively. The dashed line shows the combined effect of the three selection mechanisms.}
    \label{fig1}
\end{figure}

Figure~\ref{fig2} shows the effect of varying the fitness advantage in the
mate-choice scheme. Here \(S_1=1+\varepsilon\) and \(S_2=1\), with
\(\varepsilon=0.02,0.05,\) and \(0.10\). Panels (a), (b), and (c) correspond to
\(\alpha_1=\alpha_2=2\) above the neutral critical value,
\(\alpha_1=\alpha_2=1\) at the neutral critical value, and
\(\alpha_1=\alpha_2=0.5\) below the neutral critical value, respectively. Similar plots can be produced for the other selection schemes. Across all three effective mutation regimes, increasing \(\varepsilon\) shifts probability mass toward larger values of \(k\), reflecting the growing advantage of allele~1. Thus the effect of selection is visible not only at the critical value but also
below and above it; the magnitude of the shift increases with the fitness advantage.

\begin{figure}[htb!]
    \centering
    \includegraphics[width=0.32\textwidth]{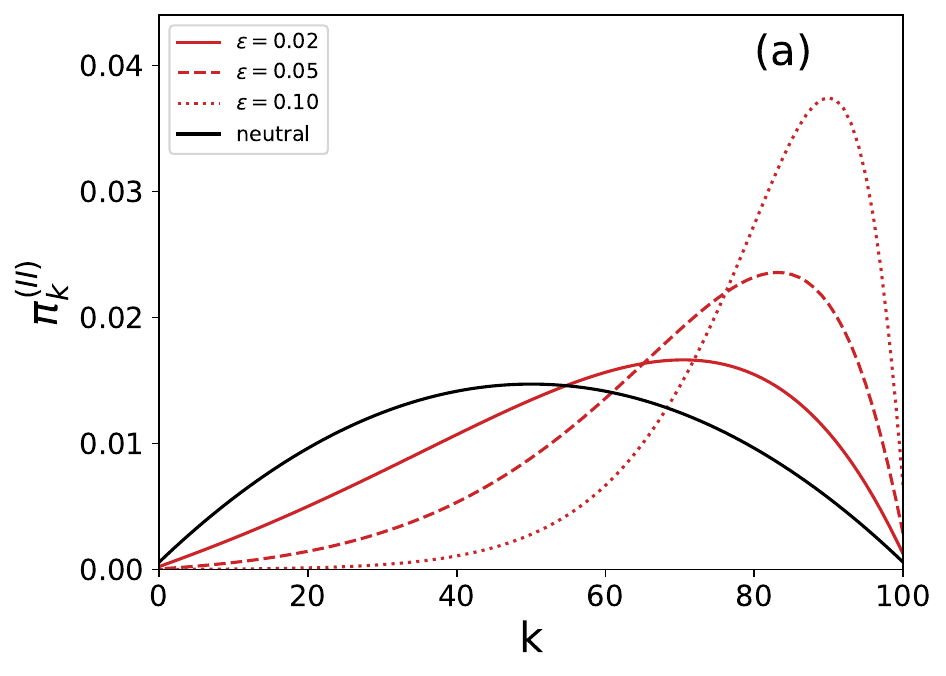}
    \includegraphics[width=0.32\textwidth]{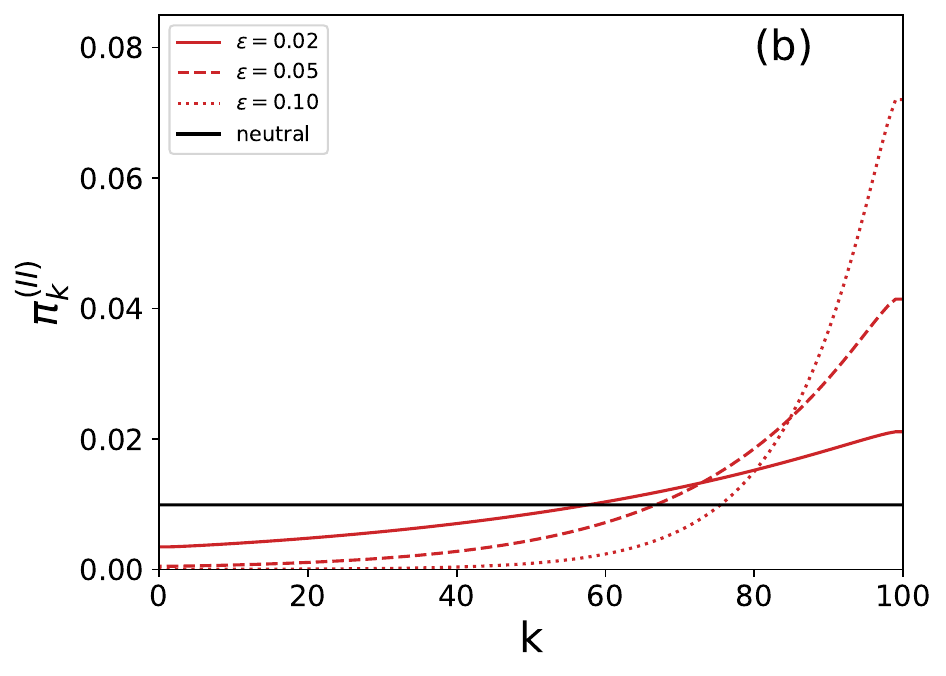}
    \includegraphics[width=0.32\textwidth]{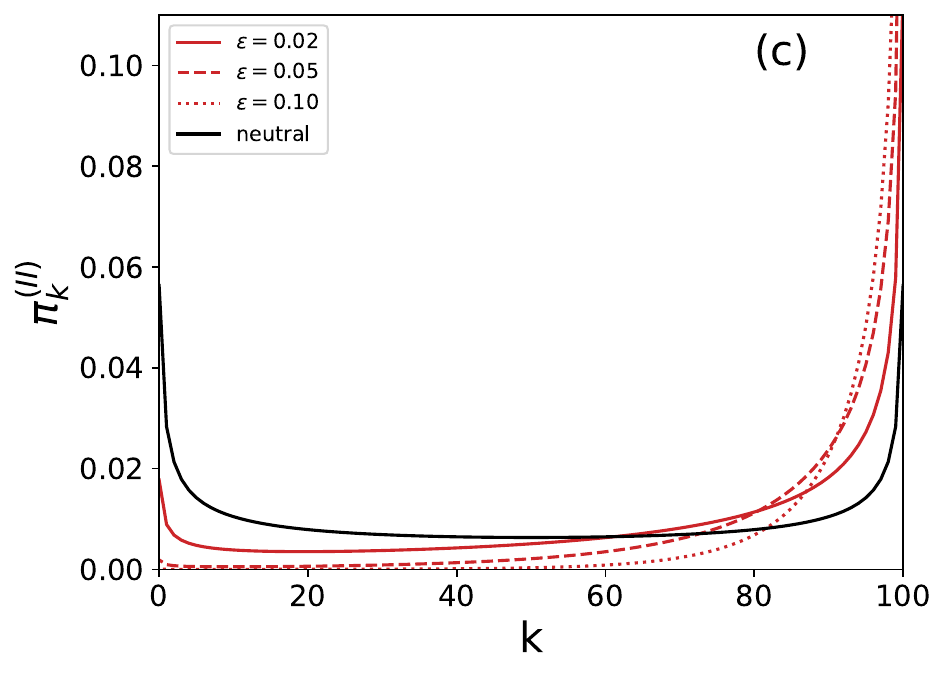}
    \caption{Probability \(\pi_k^{(II)}\) of finding \(k\) copies of allele~1 in a population of size \(n=100\) under selection during mate choice. In all panels \(S_1=1+\varepsilon\) and \(S_2=1\), with
        \(\varepsilon=0.02,0.05,\) and \(0.10\), as indicated. Panel (a) uses \(\alpha_1=\alpha_2=2\), above the neutral critical value; panel (b) uses \(\alpha_1=\alpha_2=1\), at the neutral critical value; and panel (c) uses \(\alpha_1=\alpha_2=0.5\), below the neutral critical value. The black line shows the corresponding neutral case \(S_1=S_2=1\).}
    \label{fig2}
\end{figure}

\FloatBarrier

\section{Discussion}\label{sec:discussion}

This paper studies exact stationary mutation--selection balance in a finite, well-mixed Moran framework. Its emphasis is deliberately different from the large literature on fixation probabilities, fixation times, and diffusion limits. The central question is how the stationary distribution changes when selection enters at different biological stages of a Moran-type process.

The results show that the stage at which selection acts matters. In Scheme~I, selection affects which parental allele is copied once two parents have been sampled. In Scheme~II, selection changes the sampling rule for choosing the second parent. In Scheme~III, selection changes which individual is chosen for death and replacement. These mechanisms all favor fitter alleles in a
qualitative sense, but they do not generate the same stationary distributions over allele counts.

For Schemes~I and~II, the two-allele models are exactly solvable because they are finite birth--death chains. The closed formulas in this case follow from one-dimensional detailed balance. However, this one-dimensional solvability does not imply an analogous product-form stationary distribution for three or more alleles. When \(m\ge 3\) and the fitnesses are not all equal, Schemes~I and~II
generally violate Kolmogorov's cycle criterion. That is, their stationary laws remain well defined because the finite chains are irreducible and aperiodic, but they are not generally obtained by solving detailed-balance equations. In such cases, the stationary distribution must be characterized by general finite Markov chain methods, such as the Markov chain tree representation, or computed
directly from the transition matrix. Scheme~III belongs to a different class: because selection acts through the death step, the multiallelic process remains reversible after the state-dependent factor \(\nu(k)\) is included in the stationary weight. The resulting law has a Dirichlet--multinomial core modified by explicit fitness-dependent factors.

The combined two-allele model reinforces the same message. Death selection, fitness-biased mate choice, and fitness-biased offspring copying can act simultaneously and still yield a closed stationary law when the state space is one-dimensional. The factors in the stationary weight reflect the different selection mechanisms: biased death contributes the \(\nu_k\) factor, fitness-biased mate choice contributes the product of linear factors, and fitness-biased copying contributes the explicit fitness ratio. At the same
time, this solvability, which is a consequence of birth--death structure, should not be expected to extend automatically to three or more alleles.

The weak selection calculations provide an interpretable bridge between the exact formulas and the broader weak selection literature. They show how a small fitness advantage tilts the neutral beta--binomial or Dirichlet--multinomial benchmark under each mechanism. The expansions are therefore explanatory rather
than foundational: the exact formulas remain the main objects, while the weak selection forms clarify the direction and mechanism of the departure from neutrality for small fitness advantages.

The numerical illustrations reinforce these analytical conclusions. They show that the stage at which selection acts can produce visibly different stationary distributions, even when the same allele is favored in each case. In the above-critical effective mutation regime, the single-mechanism curves separate most clearly in the upper tail, while the combined mechanism produces a much larger shift because the three selection mechanisms act in the same direction. The figures also show that increasing the fitness advantage shifts probability mass toward larger counts of the fitter allele across below-critical, critical, and above-critical effective mutation regimes.

Overall, the paper extends the classical Moran model by giving exact
finite-population stationary distributions for several biologically distinct selection mechanisms. It clarifies when the neutral Dirichlet--multinomial product-form structure can be retained after selection is added, and when selection instead leads to a genuinely nonreversible multiallelic process. These results emphasize that the stationary consequences of selection depend on where selection enters the finite-population update. The biological stage at which selection acts can determine whether detailed balance persists, whether a closed product form is available, and whether the
multiallelic dynamics are reversible or nonreversible.

\appendix

\section{Finite-state background and the neutral Dirichlet--multinomial law}\label{app:finite-neutral}
This appendix collects two finite-state facts used in the main text.

\subsection{The neutral Dirichlet--multinomial law}
\label{app:neutral-DM}
If $S_1=\cdots=S_m$, then Schemes~I and~II reduce to the neutral two-parent Moran transition kernel
\begin{equation}\label{eq:neutral-kernel-app}
Q^{(0)}\bigl(k,k+e_i-e_j\bigr)
=\frac{k_j}{n}\left(u_i+b k_i\right)
=\frac{k_j}{n}\,b(k_i+\alpha_i),
\qquad i\neq j,
\end{equation}
where $b$ and $\alpha_i$ are defined in \eqref{eq:b-alpha-def}. Scheme~III has the same conditional neutral reproduction term $T^{(0)}_{i\mid j}(k)=u_i+b k_i$; when the fitnesses are equal, its death-selection prefactor also reduces to $k_j/n$, so the full transition kernel is again \eqref{eq:neutral-kernel-app}.

The reversible stationary distribution for the neutral kernel \eqref{eq:neutral-kernel-app} is the Dirichlet--multinomial law
\begin{equation}\label{eq:neutral-DM-app}
\pi^{(0)}(k)
=\frac{n!}{(\alpha_0)_n}\prod_{i=1}^m\frac{(\alpha_i)_{k_i}}{k_i!},
\qquad
\alpha_0:=\sum_{i=1}^m\alpha_i.
\end{equation}
Indeed, for $k'=k+e_i-e_j$,
\[
\frac{Q^{(0)}(k,k')}{Q^{(0)}(k',k)}
=\frac{k_j}{k_i+1}\frac{k_i+\alpha_i}{k_j-1+\alpha_j},
\]
which is exactly the ratio of the weights in \eqref{eq:neutral-DM-app}.

\subsection{Existence and finite-state representation of the stationary law}
\label{app:finite-state}
Assume $u_i>0$ for all $i$. For each of the three schemes, every adjacent move $k\to k+e_i-e_j$ with $k_j>0$ has positive probability. Thus any state in $\state$ can be reached from any other state by a finite sequence of single replacements, and the chain is irreducible. There are also self-loops with positive probability, because a replacement attempt can leave the allele-count vector unchanged. Hence the chain is aperiodic. Since the state space is finite, each scheme has a unique stationary distribution.

For any finite irreducible Markov chain with transition matrix $Q$, the Markov chain tree theorem gives an exact stationary formula. Form the directed state graph whose vertices are the states in $\state$ and whose directed edges are the transitions $h\to h'$, $h\neq h'$, with $Q(h,h')>0$. For a fixed root $k\in\state$, a \emph{directed spanning in-tree rooted at $k$} is a directed tree that contains every state, has exactly one outgoing edge from each state $h\neq k$, has no outgoing edge from $k$, and is such that following the directed edges from any state eventually leads to $k$. Equivalently, all edges point along unique directed paths toward the root $k$. Let $\mathcal T_k$ denote the set of all such in-trees. Then
\begin{equation}\label{eq:tree-formula-app}
\pi(k)=\frac{\tau(k)}{\sum_{\ell\in\state}\tau(\ell)},
\qquad
\tau(k):=\sum_{T\in\mathcal T_k}\prod_{(h\to h')\in T}Q(h,h').
\end{equation}
This formula applies whether or not the chain is reversible. Thus, when detailed balance does not yield a product-form stationary law, the tree representation still gives an exact finite-state characterization of the stationary distribution.

\section{Proofs of nonreversibility for Schemes I and II}\label{app:nonreversibility}
We use Kolmogorov's cycle criterion. For a finite irreducible Markov chain with transition probabilities $Q(x,y)$, reversibility is equivalent to the following condition: for every directed cycle
\[
x_0\to x_1\to\cdots\to x_{r-1}\to x_0,
\]
the product of transition probabilities around the cycle must equal the product around the reversed cycle:
\[
Q(x_0,x_1)Q(x_1,x_2)\cdots Q(x_{r-1},x_0)
=
Q(x_0,x_{r-1})Q(x_{r-1},x_{r-2})\cdots Q(x_1,x_0).
\]
Therefore, to prove nonreversibility, it is enough to find one cycle for which these two products differ.

\subsection{Scheme I}
Assume $m\ge 3$, $n\ge 3$, and choose three distinct alleles $i,j,\ell$ with $S_j\neq S_\ell$. Choose a state $A=k$ with
\[
k_i=k_j=k_\ell=1,
\]
and all other counts arbitrary. Define
\[
B:=A+e_i-e_j,
\qquad
C:=A+e_i-e_\ell.
\]
These three states form the cycle
\[
A\to B\to C\to A.
\]
Using Proposition~\ref{prop:scheme-I-kernel},
\[
Q^{(I)}(A,B)=\frac{1}{n}\bigl(u_i+b s_{ij}\bigr),
\qquad
Q^{(I)}(A,C)=\frac{1}{n}\bigl(u_i+b s_{i\ell}\bigr).
\]
The other transitions in the cycle involve target alleles that are absent in the source state, so the selection correction vanishes and only the mutation term remains:
\[
Q^{(I)}(B,C)=\frac{u_j}{n},
\qquad
Q^{(I)}(C,A)=\frac{2u_\ell}{n},
\]
while in the reverse direction
\[
Q^{(I)}(C,B)=\frac{u_\ell}{n},
\qquad
Q^{(I)}(B,A)=\frac{2u_j}{n}.
\]
Therefore
\begin{align*}
\frac{Q^{(I)}(A,B)Q^{(I)}(B,C)Q^{(I)}(C,A)}
{Q^{(I)}(A,C)Q^{(I)}(C,B)Q^{(I)}(B,A)}
&=
\frac{u_i+b s_{ij}}{u_i+b s_{i\ell}}.
\end{align*}
Since
\[
s_{ij}=\frac{2S_i}{S_i+S_j},
\qquad
s_{i\ell}=\frac{2S_i}{S_i+S_\ell},
\]
and $S_j\neq S_\ell$, the ratio is not equal to one. Hence Scheme~I violates Kolmogorov's cycle criterion and is not reversible.

\subsection{Scheme II}
The proof for Scheme~II uses the same cycle. Let
\[
R:=\sum_{r\notin\{i,j,\ell\}}S_r k_r.
\]
At state $A$, the fitness-weighted total is
\[
\xi(A)=R+S_i+S_j+S_\ell.
\]
By Proposition~\ref{prop:scheme-II-kernel},
\[
Q^{(II)}(A,B)=\frac{1}{n}\left(u_i+\frac{(1-\Lambda)S_i}{2(R+S_i+S_\ell)}\right),
\]
because the focal type being replaced is $j$, so the denominator is $\xi(A)-S_j=R+S_i+S_\ell$. Similarly,
\[
Q^{(II)}(A,C)=\frac{1}{n}\left(u_i+\frac{(1-\Lambda)S_i}{2(R+S_i+S_j)}\right).
\]
The remaining four transition probabilities are again pure mutation terms:
\[
Q^{(II)}(B,C)=\frac{u_j}{n},
\qquad
Q^{(II)}(C,A)=\frac{2u_\ell}{n},
\]
\[
Q^{(II)}(C,B)=\frac{u_\ell}{n},
\qquad
Q^{(II)}(B,A)=\frac{2u_j}{n}.
\]
Thus
\begin{align*}
&\frac{Q^{(II)}(A,B)Q^{(II)}(B,C)Q^{(II)}(C,A)}
{Q^{(II)}(A,C)Q^{(II)}(C,B)Q^{(II)}(B,A)}\\
&\qquad=
\frac{u_i+\dfrac{(1-\Lambda)S_i}{2(R+S_i+S_\ell)}}
{u_i+\dfrac{(1-\Lambda)S_i}{2(R+S_i+S_j)}}.
\end{align*}
Because $S_j\neq S_\ell$, this ratio is not equal to one. Hence Scheme~II also violates Kolmogorov's cycle criterion and is not reversible.

\section{Normalizing constants}\label{app:normalizers}

This appendix gives explicit normalizing constants for the two-allele stationary laws used in the main text and shows how these finite sums can be written compactly in hypergeometric form.

We use the standard notation
\[
{}_2F_1(a,b;c;z)=\sum_{r=0}^{\infty}
\frac{(a)_r(b)_r}{(c)_r r!}z^r,
\qquad c\notin\{0,-1,-2,\ldots\},
\]
for the Gauss hypergeometric function, where \((x)_r\) denotes the rising factorial. In the formulas below, the first parameter is \(a=-n\). In that case, \((-n)_r=0\) for \(r>n\), so the defining series terminates at \(r=n\). Thus the hypergeometric expressions in this appendix are simply compact representations of the finite normalizing sums.

\subsection{Scheme I}

From \eqref{eq:scheme-I-pi}, the finite-sum normalizer is
\begin{equation}\label{eq:ZI-finite-app}
Z_I=\sum_{k=0}^n\binom{n}{k}
(\widetilde\alpha_1)_k(\widetilde\alpha_2)_{n-k}
\left(\frac{S_1}{S_2}\right)^k .
\end{equation}
Let
\[
\theta=\frac{S_1}{S_2}.
\]
To put \eqref{eq:ZI-finite-app} in hypergeometric form, rewrite the two
\(k\)-dependent factors that do not already match the hypergeometric
coefficient. First,
\[
\binom{n}{k}=(-1)^k\frac{(-n)_k}{k!}.
\]
Second, using
\[
(\widetilde\alpha_2)_n
=(\widetilde\alpha_2)_{n-k}
(\widetilde\alpha_2+n-k)_k
\]
and
\[
(\widetilde\alpha_2+n-k)_k
=(-1)^k(1-n-\widetilde\alpha_2)_k,
\]
we obtain
\[
(\widetilde\alpha_2)_{n-k}
=
(\widetilde\alpha_2)_n
\frac{(-1)^k}{(1-n-\widetilde\alpha_2)_k}.
\]
Therefore
\[
\binom{n}{k}(\widetilde\alpha_2)_{n-k}
=
(\widetilde\alpha_2)_n
\frac{(-n)_k}{(1-n-\widetilde\alpha_2)_k\,k!}.
\]
Substituting this identity into \eqref{eq:ZI-finite-app} gives
\[
Z_I
=(\widetilde\alpha_2)_n
\sum_{k=0}^n
\frac{(-n)_k(\widetilde\alpha_1)_k}
{(1-n-\widetilde\alpha_2)_k\,k!}
\theta^k .
\]
Since \((-n)_k=0\) for \(k>n\), this finite sum is exactly the terminating
Gauss hypergeometric series
\begin{equation}\label{eq:ZI-hyper-app}
Z_I=(\widetilde\alpha_2)_n\,{}_2F_1\!\left(
-n,\widetilde\alpha_1;
1-n-\widetilde\alpha_2;
\frac{S_1}{S_2}
\right).
\end{equation}

\subsection{Scheme II}

Since the product-form weights in \eqref{eq:scheme-II-product} differ from the beta-binomial-like weights in \eqref{eq:scheme-II-beta-form} by the constant factor \((\beta_2)_n^{-1}\), the corresponding product-form normalizer is obtained by multiplying the normalizer below by
\((\beta_2)_n^{-1}\).

For the beta-binomial-like form \eqref{eq:scheme-II-beta-form}, the
normalizing constant is
\begin{equation}\label{eq:ZII-finite-app}
Z_{II}
=
\sum_{k=0}^n
\binom{n}{k}
\theta_{II}^{\,k}
(\beta_1)_k(\beta_2)_{n-k}.
\end{equation}
This finite sum has the same structure as the Scheme~I normalizer
\eqref{eq:ZI-finite-app}. Applying the identity used in Appendix~C.1,
\[
\binom{n}{k}(x)_{n-k}
=
(x)_n
\frac{(-n)_k}{(1-n-x)_k\,k!},
\]
with \(x=\beta_2\), gives
\[
\binom{n}{k}(\beta_2)_{n-k}
=
(\beta_2)_n
\frac{(-n)_k}{(1-n-\beta_2)_k\,k!}.
\]
Substituting this into \eqref{eq:ZII-finite-app} yields
\[
Z_{II}
=
(\beta_2)_n
\sum_{k=0}^n
\frac{(-n)_k(\beta_1)_k}
{(1-n-\beta_2)_k\,k!}
\theta_{II}^{\,k}.
\]
Since \((-n)_k=0\) for \(k>n\), this finite sum is the terminating Gauss
hypergeometric series
\begin{equation}\label{eq:ZII-hyper-app}
Z_{II}
=
(\beta_2)_n\,
{}_2F_1\!\left(
-n,\beta_1;
1-n-\beta_2;
\theta_{II}
\right).
\end{equation}

\subsection{Scheme III}

For the two-allele Scheme~III stationary distribution
\eqref{eq:scheme-III-two-pi}, write
\[
\theta:=\frac{S_1}{S_2}.
\]
The normalizing constant is
\begin{equation}\label{eq:ZIII-two-finite-app}
Z_{III}
=\sum_{k=0}^n
\left(\frac{k}{S_1}+\frac{n-k}{S_2}\right)
\binom{n}{k}(\alpha_1)_k(\alpha_2)_{n-k}S_1^kS_2^{n-k}.
\end{equation}
This expression differs from the usual beta-binomial normalizer by the
additional linear factor
\[
\nu_k=\frac{k}{S_1}+\frac{n-k}{S_2}.
\]
To isolate this factor, define
\[
F(\theta):=\sum_{k=0}^n
\binom{n}{k}(\alpha_1)_k(\alpha_2)_{n-k}\theta^k.
\]
Since \(S_1^kS_2^{n-k}=S_2^n\theta^k\) and
\[
\frac{k}{S_1}+\frac{n-k}{S_2}
=
\frac{n}{S_2}
+k\left(\frac{1}{S_1}-\frac{1}{S_2}\right),
\]
we can rewrite \eqref{eq:ZIII-two-finite-app} as
\begin{align}\label{eq:ZIII-two-derivative-derivation-app}
Z_{III}
&=S_2^n\left[
\frac{n}{S_2}
\sum_{k=0}^n
\binom{n}{k}(\alpha_1)_k(\alpha_2)_{n-k}\theta^k
\right. \nonumber\\
&\qquad\qquad\left.
+\left(\frac{1}{S_1}-\frac{1}{S_2}\right)
\sum_{k=0}^n
k\binom{n}{k}(\alpha_1)_k(\alpha_2)_{n-k}\theta^k
\right].
\end{align}
The first sum is \(F(\theta)\), while the second is \(\theta F'(\theta)\).
Therefore
\begin{equation}\label{eq:ZIII-two-hyper-derivative-app}
Z_{III}
=
S_2^n\left[
\frac{n}{S_2}F(\theta)
+\left(\frac{1}{S_1}-\frac{1}{S_2}\right)\theta F'(\theta)
\right].
\end{equation}

It remains to write \(F\) in hypergeometric form. Using the same identity as
in Appendix~C.1,
\[
\binom{n}{k}(\alpha_2)_{n-k}
=
(\alpha_2)_n
\frac{(-n)_k}{(1-n-\alpha_2)_k\,k!},
\]
we obtain
\[
F(\theta)
=
(\alpha_2)_n
\sum_{k=0}^n
\frac{(-n)_k(\alpha_1)_k}{(1-n-\alpha_2)_k\,k!}
\theta^k.
\]
Thus
\begin{equation}\label{eq:FIII-hyper-app}
F(\theta)
=
(\alpha_2)_n\,
{}_2F_1\!\left(
-n,\alpha_1;
1-n-\alpha_2;
\theta
\right).
\end{equation}
Using the derivative identity
\[
\frac{d}{d\theta}{}_2F_1(a,b;c;\theta)
=
\frac{ab}{c}
{}_2F_1(a+1,b+1;c+1;\theta),
\]
with \(a=-n\), \(b=\alpha_1\), and \(c=1-n-\alpha_2\), gives the explicit
hypergeometric form
\begin{align}\label{eq:ZIII-two-hyper-app}
Z_{III}
&=S_2^n\Bigg[
\frac{n}{S_2}(\alpha_2)_n\,
{}_2F_1\!\left(
-n,\alpha_1;
1-n-\alpha_2;
\theta
\right) \nonumber\\
&\qquad\quad
+\left(\frac{1}{S_1}-\frac{1}{S_2}\right)
\theta(\alpha_2)_n
\frac{(-n)\alpha_1}{1-n-\alpha_2}
{}_2F_1\!\left(
1-n,\alpha_1+1;
2-n-\alpha_2;
\theta
\right)
\Bigg].
\end{align}

\section{Beta-binomial identities for the weak-selection expansion}
\label{app:weak-beta-binomial-identities}

This appendix derives the beta-binomial identities used to center the
weak-selection expansions in Section~\ref{sec:weak-selection}. The point of
the calculation is that the neutral law provides the reference distribution,
and the first-order normalization correction is an expectation with respect to
that neutral law.

For \(a>0\) and an integer \(r\ge 0\), recall that
\[
(a)_r=a(a+1)\cdots(a+r-1),
\qquad
H_r(a)=\sum_{q=0}^{r-1}\frac{1}{a+q},
\qquad H_0(a)=0.
\]
Taking logarithms gives
\[
\log (a)_r=\sum_{q=0}^{r-1}\log(a+q),
\]
and hence
\[
\frac{\partial}{\partial a}\log(a)_r=H_r(a).
\]
Therefore, by the chain rule,
\begin{equation}\label{eq:app-rising-derivative}
\frac{\partial}{\partial a}(a)_r=(a)_r H_r(a).
\end{equation}
This identity is the source of the harmonic sums in the weak-selection
formulas: a small perturbation of the first argument of a rising factorial
produces, to first order, the finite harmonic sum \(H_r(a)\).

In the two-allele neutral case, the unnormalized beta-binomial weights are
\[
W_0(k)=\binom{n}{k}(\alpha_1)_k(\alpha_2)_{n-k},
\qquad 0\le k\le n,
\]
where \(\alpha_0=\alpha_1+\alpha_2\). Vandermonde's identity gives the
normalizing constant
\begin{equation}\label{eq:app-vandermonde}
Z_0=\sum_{k=0}^n \binom{n}{k}(\alpha_1)_k(\alpha_2)_{n-k}
=(\alpha_0)_n.
\end{equation}
Thus the neutral law is the beta-binomial distribution
\begin{equation}\label{eq:app-beta-binomial-pi0}
\pi_k^{(0)}
=\frac{W_0(k)}{Z_0}
=\frac{\binom{n}{k}(\alpha_1)_k(\alpha_2)_{n-k}}{(\alpha_0)_n}.
\end{equation}
All expectations below are taken with respect to this distribution.

First differentiate \eqref{eq:app-vandermonde} with respect to \(\alpha_1\).
Using \eqref{eq:app-rising-derivative} and
\(\alpha_0=\alpha_1+\alpha_2\), the derivative of the right-hand side is
\[
\frac{\partial}{\partial\alpha_1}(\alpha_0)_n
=(\alpha_0)_n H_n(\alpha_0),
\]
while the derivative of the left-hand side is
\[
\sum_{k=0}^n
\binom{n}{k}(\alpha_1)_k(\alpha_2)_{n-k}H_k(\alpha_1).
\]
Therefore
\[
\sum_{k=0}^n
\binom{n}{k}(\alpha_1)_k(\alpha_2)_{n-k}H_k(\alpha_1)
=(\alpha_0)_n H_n(\alpha_0).
\]
Dividing by \(Z_0=(\alpha_0)_n\) and using
\eqref{eq:app-beta-binomial-pi0} yields
\begin{equation}\label{eq:app-HK-alpha1}
\mathbb E_0[H_K(\alpha_1)]=H_n(\alpha_0).
\end{equation}

Similarly, differentiating \eqref{eq:app-vandermonde} with respect to
\(\alpha_2\), again using \eqref{eq:app-rising-derivative}, gives
\[
\sum_{k=0}^n
\binom{n}{k}(\alpha_1)_k(\alpha_2)_{n-k}H_{n-k}(\alpha_2)
=(\alpha_0)_n H_n(\alpha_0).
\]
Dividing by \(Z_0=(\alpha_0)_n\) and using
\eqref{eq:app-beta-binomial-pi0} yields
\begin{equation}\label{eq:app-HnK-alpha2}
\mathbb E_0[H_{n-K}(\alpha_2)]=H_n(\alpha_0).
\end{equation}
Equations \eqref{eq:app-HK-alpha1} and \eqref{eq:app-HnK-alpha2} show why
the two harmonic terms that appear separately in the state-dependent tilt have
the same neutral expectation.

Finally, the neutral beta-binomial mean is
\begin{equation}\label{eq:app-beta-binomial-mean}
\mathbb E_0[K]=\frac{n\alpha_1}{\alpha_0}.
\end{equation}
This is the quantity denoted by \(\mu_0\) in
Section~\ref{sec:weak-selection}.

\end{document}